%% file: main.tex
\documentclass[11pt]{article}
\usepackage[T1]{fontenc}
\usepackage{graphicx}
\usepackage[margin=1in]{geometry}
\usepackage{tikz}
\usepackage{adjustbox} 
\usetikzlibrary{positioning, calc}

\usepackage[linesnumbered,ruled,vlined]{algorithm2e}

\usepackage{amsmath,amssymb,amsfonts,verbatim}
\usepackage{amsthm}
\usepackage[colorlinks,
            linkcolor=blue,
            anchorcolor=blue,
            citecolor=blue,pagebackref
            ]{hyperref}
\usepackage{array}
\usepackage{multirow}

\usepackage{todonotes}

\newtheorem{theorem}{Theorem}[section]
\newtheorem{lemma}[theorem]{Lemma}

\newtheorem{definition}[theorem]{Definition}
\newtheorem{corollary}[theorem]{Corollary}

\theoremstyle{remark} 
\newtheorem{remark}[theorem]{Remark}

\usepackage{algorithmic}
\usepackage{float}

\renewcommand{\epsilon}{\varepsilon}
\renewcommand{\le}{\leqslant}
\renewcommand{\ge}{\geqslant}

\def\CC{\mathbb{C}}

\def\F{\mathbb{F}}
\def \mC {\mathcal{C}}

\def \mC {\mathcal{C}}

\def \mF {\mathcal{F}}

\def \mL {\mathcal{L}}

\def\cL{{\mathcal L}}

\def \Xi {{X^{[i]}}}

\def \enc{\mathsf{Enc}}

\newcommand{\Ga}{\alpha}

\def \ba {{\bf a}}
\def \bb {{\bf b}}
\def \bc {{\bf c}}

\def \bx {{\bf x}}
\def\be {{\bf e}}
\def \by {{\bf y}}

\def \bu {{\bf u}}
\def \bv {{\bf v}}

\def\cC{\mathcal{C}}
\def\F{\mathbb{F}}
\def\bu{\mathbf{u}}
\def\bv{\mathbf{v}}
\def\bw{\mathbf{w}}
\def\bx{\mathbf{x}}
\def\bc{\mathbf{c}}

\def\LRM{\mathrm{Ligero_{rm}}}

\def\rk{\mathrm{rk}}
\def\isequal{\stackrel{?}{=}}
\def\Po{\mathcal{P}}
\def\V{\mathcal{V}}
\def\cA{\mathcal{A}}
\def\cR{\mathcal{R}}
\def\IOP{\mathrm{IOP}}
\def\IOPP{\mathrm{IOPP}}

\def\cC{\mathcal{C}}
\def\cG{\mathcal{G}}
\def\bc{\mathbf{c}}
\def\F{\mathbb{F}}

\input{sections/tcb}
\begin{document}
\title{Proximity Gaps for Gabidulin Codes and Applications}
%
%

\author{Songsong Li\thanks{Shanghai Jiao Tong University, Shanghai, China. \texttt{songsli@sjtu.edu.cn}.}
\and  Chaoping Xing \thanks{Shanghai Jiao Tong University, Shanghai, China. \texttt{xingcp@sjtu.edu.cn}.}
\and Chen Yuan\thanks{Shanghai Jiao Tong University, Shanghai, China. \texttt{chen\_yuan@sjtu.edu.cn}.}
\and Ruiqi Zhu\thanks{Shanghai Jiao Tong University, Shanghai, China. \texttt{sjtuzrq7777@sjtu.edu.cn}.}
}
\date{}
\maketitle

\input{sections/abstract}

%

%
%
%
\newpage
\tableofcontents
\newpage

\input{sections/intro}
\input{sections/preliminary}

\input{sections/Proximity-gap}
\input{sections/Ligero_rm}

\bibliographystyle{alpha}
\bibliography{reference}

\end{document}

%% file: sections/tcb.tex
\usepackage{tcolorbox}
\tcbuselibrary{skins,breakable}

\makeatletter
\newcommand{\DrawLine}{%
	\begin{tikzpicture}
		\path[use as bounding box] (0,0) -- (\linewidth,0);
		\draw[color=black!75,dashed,dash phase=2pt]
		(0-\kvtcb@leftlower-\kvtcb@boxsep,0)--
		(\linewidth+\kvtcb@rightlower+\kvtcb@boxsep,0);
	\end{tikzpicture}%
}
\makeatother

\newtcolorbox{mybox}[2][]{%
	enhanced,
	title        = {#2},
	attach boxed title to top left={xshift=+3mm,yshift*=-3mm},
	breakable    = true,
	colback      = black!4,
	colframe     = black!75,
	fonttitle    = \bfseries,
	colbacktitle = black!10!white,
	coltitle     = black,
	#1
}

\newtcolorbox[auto counter]{functionality}[2][]{%
	enhanced,
	title        = {Functionality~\thetcbcounter: #2},
	attach boxed title to top left={xshift=+3mm,yshift*=-3mm},
	breakable    = true,
	colback      = yellow!4,
	colframe     = black!75,
	fonttitle    = \bfseries,
	fontupper    = \small,
	fontlower    = \small,
	colbacktitle = yellow!10!white,
	coltitle     = black,
	#1
}

\newtcolorbox[use counter=pro]{protocol}[2][]{%
	enhanced,
	title        = {Protocol~\thetcbcounter: #2},
	attach boxed title to top left={xshift=+3mm,yshift*=-3mm},
	breakable    = true,
	colback      = black!4,
	colframe     = black!75,
	fonttitle    = \bfseries, 
	fontupper    = \small,
	fontlower    = \small,
	colbacktitle = black!10!white,
	coltitle     = black,
	#1
}

\newtcbox{\xmybox}[1][red]{on line,
	arc=7pt,colback=#1!10!white,colframe=#1!50!black,
	before upper={\rule[-3pt]{0pt}{10pt}},boxrule=1pt,
	boxsep=0pt,left=6pt,right=6pt,top=2pt,bottom=2pt}

%% file: sections/abstract.tex
\begin{abstract}
Proximity gaps are central to the soundness of interactive oracle proofs of proximity (IOPPs) and polynomial commitment schemes (PCSs). An $[n,k,d]$ linear code $C\subseteq\mathbb F^n$ has a $\delta$-proximity gap with error $\epsilon$ if, for every $u_0,u_1\in\mathbb F^n$, either all points on $\ell_{u_0,u_1}=\{u_0+\alpha u_1:\alpha\in\mathbb F\}$ are $\delta$-close to $C$, or at most an $\epsilon$ fraction are. Although proximity gaps for Hamming-metric codes are well understood, their rank-metric counterparts remain largely unexplored despite their applications in coding theory and cryptography. In this work, we study proximity gaps for linear rank-metric codes and their cryptographic applications.

First, we show that every $[n,k,d]$ linear rank-metric code $C$ over $\mathbb F_{q^m}$ admits a proximity gap for every $\delta\le(d-1)/(3n)$, with error at most $q^{e+1}/q^m$, where $e=\lfloor\delta n\rfloor$. For Gabidulin codes, we improve the gap to $(d-1)/(2n)$ with error $10q^{n-1}/q^m$. These two proximity gaps match those for general linear Hamming-metric codes and Reed--Solomon (RS) codes, respectively. We prove the $(d-1)/(2n)$ bound is tight by constructing an infinite family of constant-rate Gabidulin codes and affine lines $\ell_{u_0,u_1}$ on which a $1-o(1)$ fraction of points are $d/(2n)$-close to the code, while $u_1$ is at least $3d/(4n)$-far from it. At $\delta=(d-1)/(2n)$, a further counterexample gives a lower bound of order $q^{n-1}/q^m$ on $\varepsilon$, matching our upper bound up to a constant factor.

As applications, we construct an IOPP for interleaved Gabidulin codes by adapting the Ligero IOPP for interleaved RS codes. We then adapt the Ligero-based PCS for ordinary polynomials to obtain a $q$-linearized polynomial commitment scheme. To our knowledge, this is the first PCS framework based on rank-metric error-correcting codes.
\end{abstract}

%% file: sections/intro.tex
\section{Introduction}

A polynomial commitment scheme (PCS) enables a prover to succinctly commit to a polynomial and later prove its evaluations at chosen points, serving as a key building block in succinct non-interactive arguments of knowledge (SNARKs)~\cite{GabizonWC19,ChiesaHMMVW20,Setty20}, with later works such as Brakedown~\cite{brakedown} and BaseFold~\cite{basefold} further improving the efficiency and field
flexibility of polynomial commitments. Beyond SNARKs, PCSs have also found applications in verifiable secret sharing~\cite{kate2010constant,ZhangLGSCLD24}, blockchain data availability sampling~\cite{HallAndersenSW24}, and secure multiparty computation~\cite{bhadauria2023private}. An efficient class of PCS construction is based on interactive oracle proof of proximity (IOPP) for error-correction codes~\cite{ligero,fri}. IOPP is an interactive oracle proof, where the verifier queries the prover's messages as oracles to verify that a witness is close to a prescribed error-correction code. 


A key soundness ingredient in these cryptographic constructions is the \emph{proximity gap} of the underlying error-correction code. Let $\cC\subseteq \F^n$ be an $[n,k,d]$ linear code. For two vectors $\bu_0,\bu_1\in \F^n$, consider the affine line
\[
    \ell_{\bu_0,\bu_1}
    =
    \{\bu_0+\lambda \bu_1:\lambda\in \F\}.
\]
For real numbers $\delta,\epsilon\in (0,1)$, the code $\cC$ is said to have a $(\delta,\epsilon)$-proximity gap for lines if, for every such line, either all points on the line are $\delta$-close to $\cC$ under the Hamming distance, or at most an $\epsilon$-fraction of the points on the line are $\delta$-close to $\cC$. Here a vector is said to be $\delta$-close to $\cC$ if its relative Hamming distance to $\cC$ is at most $\delta$; otherwise it is said to be $\delta$-far from $\cC$. A stronger property called \emph{correlated agreement} guarantees that if more than an $\epsilon$-fraction of the points on the line are $\delta$-close to $\cC$, then there is a common agreement set $E\subseteq [n]$ of size at least $(1-\delta)n$ and two codewords $\bc_0$ and $\bc_1$ such that $\bu_0$ and $\bu_1$ agree with $\bc_0$ and $\bc_1$ at $E$, respectively, implying that all points on $\ell_{\bu_0,\bu_1}$ are $\delta$-close to $\cC$.

In the Hamming-metric setting, proximity gaps are well understood and underpin the soundness of IOPPs. Typically, larger values of the proximity parameter $\delta$  yield stronger soundness guarantees and smaller proof sizes of IOPPs, leading to efficient IOPP-based PCSs. Among various code families, Reed-Solomon (RS) codes serve as the canonical choice for IOPP-based constructions: their algebraic structure enables efficient proximity testing such as FRI~\cite{fri} and further refinements include STIR and WHIR~\cite{stir,whir}. Another influential line of work, beginning with Ligero~\cite{ligero}, uses IOPPs for interleaved RS codes and has led to several efficient variants and extensions, including Brakedown, Basefold, and related multivariate constructions~\cite{brakedown,basefold,blaze25}.
 
We briefly review the known results on proximity gaps. For general Hamming-metric linear codes of block length $n$ and minimum distance $d$, it was shown in~\cite{ligero} that they admit a proximity gap of $\delta\leq (d-1)/4n$ and improved to $\delta\leq (d-1)/3n$ in \cite{AmesHIV23}. A line of work initiated by~\cite{Ben-SassonKS18} and further developed
in~\cite{BSGKS20,GKL24} establishes proximity gaps using list-decoding techniques, achieving the ``$1.5$ Johnson bound,'' namely
$\delta<1-\sqrt[3]{1-d/n}$. For RS codes, unique-decoding and list-decoding techniques yield stronger proximity gaps, namely $\delta=(d-1)/2n$ and $\delta<1-\sqrt{1-d/n}$, respectively~\cite{BCIKS23}. 
Recent works have further developed the theory of RS proximity gaps by providing both improved positive results and counterexamples: Goyal and Guruswami obtained optimal proximity gaps for subspace-design codes and  RS codes with random evaluation points~\cite{GoyalG26}, while Ben-Sasson, Carmon, Haböck, Kopparty and Saraf gave refined positive results as well as counterexamples for RS proximity gaps~\cite{Ben-SassonCHKS26}; see also~\cite{CritesS25} for related limitations and conjectures. 

All the above results are for the Hamming metric. We seek to understand whether proximity gaps for Hamming-metric codes have analogues in other metrics. Unlike the Hamming metric, which counts the number of non-zero coordinates, the rank metric is defined as the dimension of the subspace spanned by the vector's coordinates. To the best of our knowledge, proximity gaps for rank-metric codes have not yet been studied. 

Rank-metric codes, particularly Gabidulin codes~\cite{G85}, have found wide applications in random network coding~\cite{KoetterK08,SilvaKK08}, distributed storage~\cite{BartzHLPRW22}, and post-quantum cryptography~\cite{gabidulin1991ideals}. In particular, rank-metric hardness assumptions, such as the Rank Syndrome Decoding problem~\cite{GabZem14} and MiniRank~\cite{Cou01}, underlie several cryptographic constructions, including RQC~\cite{aguilarrank} and Durandal~\cite{DBLP:conf/eurocrypt/AragonBGHZ19}. However, verifying rank-metric statements (e.g., proving that a vector over an extension field $\F_{q^m}$ or a matrix over $\F_q$ after fixing an $\F_q$-basis of $\F_{q^m}$ has low rank over $\F_q$) in existing proof systems typically requires reducing these statements to generic arithmetic circuits over the base field that will incur significant overhead, as encoding relevant matrix operations and linear dependencies over extension fields leads to large circuit blowups. A more natural and efficient approach would be to construct proof systems natively in the rank-metric setting. 

Motivated by the above theoretical and practical considerations, we study proximity gaps for linear rank-metric codes and leverage them to construct IOPPs for rank-metric codes and corresponding polynomial commitment schemes.


\subsection{Our Contributions}

We first formalize the definition of proximity gaps under the rank-metric. Throughout, we work over an extension field $\F_{q^m}$ and assume $n\leq m$. Fixing an $\F_q$-basis of $\F_{q^m}$ identifies each vector in $\F_{q^m}^n$ with an $m\times n$ matrix over $\F_q$. For any $\bx,\by\in\F_{q^m}^n$, we define their rank distance and relative rank distance as
\[
    d_{\mathrm{rk}}(x,y):={\operatorname{rank}_{\F_q}(X-Y)},\quad \Delta_{\mathrm{rk}}(\bx,\by)
    =d_{\mathrm{rk}}(x,y)/n,
\]
where $X,Y\in \F_q^{m\times n}$ are the matrix representations of $\bx,\by$. Let $\cC\subseteq \F_{q^m}^n$ be a rank-metric code. The (relative) rank-distance between $\bx\in\F_{q^m}^n$ and $\cC$ is defined by $d_{\mathrm{rk}}(\bx,\cC)=\min_{\bc\in \cC}d_{\mathrm{rk}}(\bx,\bc)$ (resp. $\Delta_{\mathrm{rk}}(\bx,\cC)=d_{\mathrm{rk}}(\bx,\cC)/n$). %

For any two vectors $\bu_0,\bu_1\in\F_{q^m}^n$, we define their order-$2$ interleaving as
\[
    [\bu_0,\bu_1]
    :=
    \begin{pmatrix}
        \bu_0\\
        \bu_1
    \end{pmatrix}
    \in \F_{q^m}^{2\times n}.
\]
The rank of $[\bu_0,\bu_1]$ is defined as the rank of the $2m\times n$ matrix obtained by expanding each entry in the fixed $\F_q$-basis of $\F_{q^m}$. For a code $\cC\subseteq\F_{q^m}^n$, its interleaved code of order $2$ is $\cC^2= \{[\bc_0,\bc_1] : \bc_0,\bc_1\in\cC\}$. 

\begin{definition}[Rank-metric proximity gap for lines]
Let $\cC \subseteq \F_{q^m}^n$ be a linear rank-metric code. We say that $\cC$ has a $(\delta,\epsilon)$-\emph{proximity gap for lines} if, for all $\bu_0,\bu_1 \in \F_{q^m}^n$, the condition
\[
    \Pr_{\lambda \leftarrow \F_{q^m}}
    \left[
        \Delta_{\mathrm{rk}}(\bu_0+\lambda\bu_1,\cC)\le \delta
    \right]
    >
    \epsilon
\]
implies that
\[
    \Delta_{\mathrm{rk}}([\bu_0,\bu_1],\cC^2)\le \delta.
\]
Equivalently, there exist codewords $\bc_0,\bc_1\in\cC$ such that the interleaved vector $[\bu_0-\bc_0,\bu_1-\bc_1]$ has rank at most $\delta n$. Typically, $\delta$ is called the proximity gap parameter, and the fraction $\epsilon$ is called the soundness error parameter.
\end{definition}

Our first result shows that the proximity-gap phenomenon up to $\frac{d-1}{3n}$ also holds for arbitrary linear rank-metric codes.

\begin{theorem}[Proximity gaps for linear rank-metric codes]\label{thm:main-metric-codes}
Let $\cC\subseteq \F_{q^m}^n$ be an $\F_{q^m}$-linear rank-metric code with minimum rank distance $d$. For $0<\delta\le (d-1)/3n$, let $e=\lfloor \delta n\rfloor$. If
\[
    \Pr_{\lambda\leftarrow\F_{q^m}}
    \left[
        \Delta_{\mathrm{rk}}(\bu_0+\lambda\bu_1,\cC)\le \delta
    \right] > q^{e+1}/q^m,
\]
then
\[
    \Delta_{\mathrm{rk}}([\bu_0,\bu_1],\cC^2)\le \delta.
\]
Consequently, either all points on the affine line are $\delta$-close to $\cC$, or only a small fraction $q^{e+1}/q^m$ of them are.
\end{theorem}

We then obtain a stronger proximity-gap for Gabidulin codes. A Gabidulin code $\cC=\mathcal G(\mathcal{A},k)$ is obtained by evaluating all $q$-linearized polynomials of $q$-degree less than $k$ on a set $\mathcal A=\{\alpha_1,\ldots,\alpha_n\}\subseteq \F_{q^m}$ that is $\F_q$-linearly independent, i.e.,
\[
\mathcal G(\mathcal{A},k):=\left\{f(\cA)=(f(\Ga_1),\dots,f(\Ga_n)):\ f=\sum_{i=0}^{k-1}a_iX^{q^i},\ a_i\in\F_{q^m}\right\}.
\]
It is viewed as the rank-metric analogue of Reed--Solomon codes. By using rank-metric Berlekamp--Welch decoding algorithm, we improve the proximity parameter $\delta$ from $\frac{d-1}{3n}$ to $\frac{d-1}{2n}$, matching the unique-decoding proximity bound known for Reed--Solomon codes in the Hamming metric.

\begin{theorem}[Proximity gaps for Gabidulin codes]\label{thm:main2}
Let $\cC\subseteq\F_{q^m}^n$ be an $[n,k,d]$ Gabidulin code. For every integer $r\ge0$, define $\sigma_r:=\sum_{i=0}^{r}q^i$. For $0< \delta\le \frac{d-1}{2n}$, let $e=\lfloor \delta n\rfloor$. Then $\cC$ admits $(\delta,\epsilon_e)$-proximity gap with $$\epsilon_e:=\frac{\tau_e}{q^m}\leq \frac{10q^{n-1}}{q^m},$$
    where $\tau_e:=(\sigma_e-1)+\sigma_e((\sigma_e-1)\sigma_{k-1}+1)\leq 10q^{n-1}.$
\end{theorem}

We complement the positive result above with two counterexamples. Our first counterexample is inspired by the counterexample in \cite{Ben-SassonCHKS26} which characterizes the proximity gaps of Reed--Solomon codes. It shows the tightness of the $\frac{d-1}{2n}$-gap for Gabidulin codes: at radius $\delta=\frac{d}{2n}$, there is a family of Gabidulin codes $\cC$ and affine lines for which although all points are $\delta$-close to $\cC$, the interleaved word $[\bu_0,\bu_1]$ remains $\delta$-far from $\cC^2$. This can be compared with the strong list-decoding lower bounds just beyond half-the-distance for certain Gabidulin codes~\cite{RavivW16,RavivW17}, though the construction and analysis underlying Theorem~\ref{thm:main-radius-counterexample} are different.

\begin{theorem}[Counterexample of Gabidulin codes at proximity gap $\frac{d}{2n}$]\label{thm:main-radius-counterexample}
For every prime power $q$, there is an infinite family of Gabidulin codes indexed by $g\ge 1$ as follows. Let $n=6g$, and let $\cC\subseteq \mathbb F_{q^n}^n$ be an $[n,2g+1,4g]$ Gabidulin code. Then, there exist vectors $\bu_0,\bu_1\in \mathbb F_{q^n}^n$ such that
\[
\Pr_{\lambda\in \mathbb F_{q^n}} \left[d_{\mathrm{rk}}(\bu_0+\lambda \bu_1,\cC)\le d/2\right]=1,
\]
but $d_{\rk}([\bu_0,\bu_1],\cC^2)\ge d_{\rk}(\bu_1,\cC)\ge \frac{3d}{4}>d/2$.
\end{theorem}

Exchanging the two directions gives a line $\bu_1+\lambda\bu_0$ with exactly $q^n-1$ points within distance $d/2$ of $\cC$, while its point at $\lambda=0$ remains at distance at least $3d/4$. Thus the weaker all-or-few proximity-gap property requires soundness error at least $1-q^{-n}$.

Section~\ref{sec:limitations-radius} extends this construction to $n=Ng$ and $k=(N-4)g+1$ for every integer $N\ge6$. The fraction of close points on $\bu_0+\lambda\bu_1$ is at least $1-q^{-n}$, and equals $1$ when $N$ is even. These counterexamples can have rates arbitrarily close to one.

Our second counterexample shows that, at $\delta=\frac{d-1}{2n}$, the $q^{n-1}/q^m$ dependence of the soundness error in Theorem~\ref{thm:main2} is tight up to a constant factor.
\begin{theorem}[A tight lower bound on $\epsilon$ at proximity gap $\frac{d-1}{2n}$]\label{thm:main-half-threshold-counterexample}
For every prime power $q$ and integers $r\ge1$ and $2\le c\le h$, let $n=c(r+1)$, $m=h(r+1)$, and $e=r$. Then there exist an $[n,n-2r,2r+1]$ Gabidulin code $\cC\subseteq\F_{q^m}^n$ and vectors $\bu_0,\bu_1\in\F_{q^m}^n$ such that
\[
\Pr_{\lambda\in\F_{q^m}}\left[d_{\mathrm{rk}}(\bu_0+\lambda\bu_1,\cC)\le e\right]
\ge\frac{q^n-1}{(q-1)q^m}
=\Theta\!\left(\frac{q^{n-1}}{q^m}\right),
\]
but $d_{\mathrm{rk}}([\bu_0,\bu_1],\cC^2)\ge d_{\mathrm{rk}}(\bu_1,\cC)\ge e+1$.
\end{theorem}

Exchanging $\mathbf u_0$ and $\mathbf u_1$ gives a line with at least $(q^n-1)/(q-1)$ points within distance $e$ of $\cC$, but also a point farther than $e$. Thus, the same lower bound on $\varepsilon$ holds even if we only require that a line with more than an $\varepsilon$ fraction of close points must have all its points close.


Recall that a succinct PCS allows a verifier check a committed polynomial has the claimed degree bound and evaluation with sublinear complexity measured by its input size. As an application of our rank-metric proximity-gap results, we construct a PCS for $q$-linearized polynomials. This differs from the Hamming-metric setting, where Reed--Solomon codes are defined by evaluating ordinary low-degree polynomials, and their proximity properties are used to construct efficient PCSs for ordinary polynomials. In the rank-metric setting, the corresponding error-correction codes are Gabidulin codes, which are defined by evaluating $q$-linearized polynomials $f(X)=\sum_{i=0}^{k-1}a_iX^{q^i}$ of a bounded $q$-degree. A new PCS tailored to $q$-linearized polynomials is needed because their ordinary degree can be much larger than their input size. In particular, the polynomial above has only $k$ coefficients but has ordinary degree as large as $q^{k-1}$. Therefore, directly applying a PCS for ordinary polynomials may give complexity that is sublinear in $q^{k-1}$, but not in the input size $k$. Our goal is instead to obtain a PCS whose complexity is sublinear in $k$ by exploiting the structure of $q$-linearized polynomials and the proximity gap of Gabidulin codes. 

In the Hamming-metric setting, there are two major families of IOPP-based PCS constructions. One is FRI-based schemes that rely on FFT-friendly evaluation domains, another is Ligero-based schemes that use IOPPs for interleaved RS codes. Since no comparable FFT is available for $q$-linearized polynomials, we follow the Ligero framework and construct a $q$-linearized PCS based on interleaved Gabidulin codes, which we call $\LRM$-PCS.


\begin{theorem}[$\LRM$-PCS, informal]\label{thm:main3}
$\LRM$-PCS~\ref{proc:lcom} is a transparent and extractable polynomial commitment scheme for $q$-linearized polynomials over $\F_{q^m}$ of $q$-degree less than $k$. For security parameter $\lambda$, it has prover time $\widetilde O(\lambda k)$, verifier time $\widetilde O(\lambda\sqrt{k})$, and proof size $\widetilde O(\lambda\sqrt{k})$. Here, running times are measured in
operations over $\F_{q^m}$, and $\widetilde O(\cdot)$ suppresses
polylogarithmic factors in $k$.
\end{theorem}

\subsection{Technical Overview}\label{subsec:technical-overview}

\paragraph{Proximity gaps for Gabidulin codes.} We focus on the main ideas behind the $(d-1)/(2n)$-proximity gap for Gabidulin codes. The proof of the $(d-1)/(3n)$-proximity gap for general linear rank-metric codes adapts the corresponding Hamming-metric argument and is therefore omitted from this overview. The starting point is similar to the Reed--Solomon case: if many points on the affine line $\ell_{\bu_0,\bu_1}:=\{\bu_0+z\bu_1:z\in\mathbb F_{q^m}\}$ are close to a Gabidulin code, then for each such good point $z$ one can write down a Berlekamp--Welch (BW) reconstruction system. Since different values of $z$ may give different decoded $q$-linearized polynomials, different error spaces, and different error-locator polynomials, we need to show that, once there are sufficiently many good points, these pointwise systems can all be viewed as one linear system over the rational function field $\mathbb F_{q^m}(Z)$, leading to a global solution. In the RS case, this step is guaranteed by the Polishchuk--Spielman (PS) divisibility lemma~\cite{PS94,BCIKS23} over the ordinary bivariate polynomial ring. However, in the Gabidulin case, the objects are $q$-linearized polynomials, whose multiplication is symbolic product $``\circ"$ (i.e., composition, see~Section~\ref{sec:linearizedpoly} for details) rather than ordinary multiplication. The PS lemma cannot be applied in this setting. This is where the argument for Gabidulin code departs from its Reed--Solomon counterpart.

Let $n\leq m$ be positive integers and let $\cA=\{\Ga_1,\dots,\Ga_n\}$ be a set of $n$ $\F_q$-linearly independent elements in the finite field $\F_{q^m}$. For every $q$-linearized polynomial $f(X)$, we denote its evaluations on $\cA$ by $f(\cA)=(f(\Ga_1),\dots,f(\Ga_n))$. Let $\cC=\mathcal{G}(\mathcal{A},k)$ be an $[n,k]$ Gabidulin code with minimum distance $d=n-k+1$ (see~Section~\ref{sec:2}). Let $e\leq (d-1)/2$ and $S:=\{z\in\mathbb F_{q^m}:d_{rk}(\bu_0+z\bu_1,\cC)\le e\}$ be the set of good points on the line. In the following, we show that once $|S|>\tau_e$, here $\tau_e$ is defined as in  {Theorem~\ref{thm:main2}}, there exist $\bc_0,\bc_1\in \cC$ such that $d_{\mathrm{rk}}\left([\bu_0-\bc_0],[\bu_1-\bc_1]\right)\leq e$.

For every $z\in S$, since we are within the unique decoding radius, there is a unique $q$-linearized polynomial $P_z$ of $q$-degree less than $k$ such that $d_{\mathrm{rk}}\bigl(\bu_0+z\bu_1,P_z(\mathcal{A})\bigr)\le e$. Moreover this $P_z$ can be solved by the BW algorithm for Gabidulin codes. Precisely, given $\bu_0+z\bu_1$, let $W_z$ denote the $\mathbb F_q$-span of the error values $\mathbf{u}_{0,r}+z\mathbf{u}_{1,r}-P_z(\alpha_r),\,r\in [n]$ and let $\Lambda_z(X)$ be the subspace-annihilator polynomial of $W_z$. Since $\Lambda_z(X)$ has no multiple roots, the coefficient of $X$ is nonzero. We may therefore normalize this coefficient to be $1$ and write $
\Lambda_z(X)=X+\sum_{i=1}^e \lambda_i(z)X^{q^i}.$ Define $\Omega_z(X):=-\Lambda_z(P_z(X))=\sum_{j=0}^{k+e-1}\omega_j(z)X^{q^j}.$ Then we have the pointwise BW equation $$\Lambda_z(u_{0,r}+zu_{1,r})+\Omega_z(\alpha_r)=0,\ r\in[n].$$ Equivalently, given $M(z)$ and $c(z)$ defined as follows, the linear system~\eqref{eq:bweq} in the unknown coefficients of the
linearized polynomials $A_z=X+\sum_{i=1}^ea_i(z)X^{q^i}$ and $B_z=\sum_{j=0}^{e+k-1}b_j(z)X^{q^j}$ is solvable, since $(A_z=\Lambda_z,B_z=\Omega_z)$ gives one solution:
\begin{equation}\label{eq:bweq}
\small\underbrace{
\begin{pmatrix}
(u_{0,1}+zu_{1,1})^q
& \cdots
& (u_{0,1}+zu_{1,1})^{q^e}
& \alpha_1
& \cdots
& \alpha_1^{q^{k+e-1}}
\\
(u_{0,2}+zu_{1,2})^q
& \cdots
& (u_{0,2}+zu_{1,2})^{q^e}
& \alpha_2
& \cdots
& \alpha_2^{q^{k+e-1}}
\\
\vdots
& & \vdots
& \vdots
& & \vdots
\\
(u_{0,n}+zu_{1,n})^q
& \cdots
& (u_{0,n}+zu_{1,n})^{q^e}
& \alpha_n
& \cdots
& \alpha_n^{q^{k+e-1}}
\end{pmatrix}}_{:=M(z)}
\begin{pmatrix}
a_1(z)\\
\vdots\\
a_e(z)\\
b_0(z)\\
\vdots\\
b_{k+e-1}(z)
\end{pmatrix}
=
-\underbrace{
\begin{pmatrix}
u_{0,1}+zu_{1,1}\\
u_{0,2}+zu_{1,2}\\
\vdots\\
u_{0,n}+zu_{1,n}
\end{pmatrix}}_{:=c(z)}.
\end{equation}
Indeed, for any pair of solutions $(A_z,B_z)$, we have $A_z \circ P_z = -B_z.$ Thus each $P_z$ can be obtained by right-dividing $B_z$ by $A_z$ with respect to the symbolic product. 

To give a global solution, instead we view these pointwise systems as specializations of a single linear system over the rational function field $\mathbb F_{q^m}(Z)$, i.e., solving for unknowns $U:=(a_1(Z),\dots,b_{k+e-1}(Z))\in\F_{q^m}(Z)^{2e+k}$ such that
\[
M(Z)U=-c(Z).
\]
We next show that this system is still solvable over $\mathbb F_{q^m}(Z)$, equivalently,
\[
\operatorname{rank}_{\mathbb F_{q^m}(Z)}\left(M(Z)\mid c(Z)\right)
=
\operatorname{rank}_{\mathbb F_{q^m}(Z)}M(Z).
\]
This follows from two observations. First, for each $z\in S$, the specialized system $M(z)U=-c(z)$ is solvable over $\F_{q^m}$, and hence the corresponding two ranks agree after specialization at $z$. Second, if $\operatorname{rank}_{\mathbb F_{q^m}(Z)}M(Z)=r$, any $r+1$ minor of the augmented matrix $(M(Z)\mid c(Z))$ has degree at most $\sigma_e$ in $Z$. Therefore, if $|S|>\sigma_e$, any such minor that vanishes on all points of $S$ must vanish identically. Hence $M(Z)U=-c(Z)$ has a solution over $\mathbb F_{q^m}(Z)$ (see~Lemma~\ref{lem:bw-interpolation} for details). Clearing denominators in such a rational solution, we assume $a_i(Z), b_i(Z)\in\F_{q^m}[Z]$ and $a_0(Z)\not\equiv 0$. Moreover, Lemma~\ref{lem:bw-interpolation} gives the degree bounds $\deg_Z a_i\le \sigma_e-q^i$ for $0\leq i\leq e$ and $\deg_Z b_j\le \sigma_e$ for $0\leq j\le k+e-1$. Set
\begin{equation}\label{eq:solus}
A(X,Z)=\sum_{i=0}^e a_i(Z)X^{q^i},
\qquad
B(X,Z)=\sum_{j=0}^{k+e-1}b_j(Z)X^{q^j},
\end{equation}
then $$A(u_{0,r}+Zu_{1,r},Z)+B(\alpha_r,Z)=0,\ r\in[n].$$

The next step is where the Gabidulin argument departs from the RS case. In the RS setting, the analogous construction produces two polynomials $A(X,Z)$ and $B(X,Z)$ in the bivariate polynomial ring. We then have that $A(X,z)$ divides $ B(X,z)$ for sufficiently many values of $z$,  and $A(x,Z)$ divides $B(x,Z)$ for all points $x$ in the evaluation domain. Then, by the Polishchuk--Spielman bivariate divisibility lemma, $A(X,Z)$ divides $ B(X,Z)$ under the ordinary bivariate polynomial multiplication. Writing $P(X,Z)=-B(X,Z)/A(X,Z)$, one then uses the degree bounds on $A$ and $B$ to show that $P=v_0(X)+Zv_1(X)$, where $v_0,v_1$ both have degree at most $k-1$. Thus, the quotient $P$ packages the pointwise decoded polynomials into a single low-degree polynomial, from which the desired correlated agreement follows. For Gabidulin codes, however, the identity $A_z\circ P_z=-B_z$ involves symbolic product rather than ordinary multiplication, so the PS lemma does not apply. We therefore use a different approach, replacing the divisibility argument with a coefficient-by-coefficient reconstruction of $P(X,Z)$ and showing that it is linear in $Z$.

Let $A(X,Z)$ and $B(X,Z)$ be the solution given in Equation~\eqref{eq:solus}. For each $z\in S$, the pair $A(X,z)$ and $B(X,z)$ is also a solution of Equation~\eqref{eq:bweq}. Then there exists $q$-linearized polynomial $P_z(x)$ of $q$-degree less than $k$ satisfying
\[A(P_z(X),z)=-B(X,z).\] 
From all $\{P_z(X):\ z\in S\}$, one can interpolate a bivariate polynomial $P(X,Z)\in\F_{q^m}[X,Z]$ such that $P(X,z)=P_z(X)$ for every $z\in S$. Let $S(Z):=\prod_{z\in S}(Z-z)$. Then 
\begin{equation}\label{eq:compositionmodS}
    A(P(X,Z),Z)+B(X,Z)\equiv 0 \pmod{S(Z)}.
\end{equation}
Next, we will show that the $Z$-degree of $P(X,Z)$ is actually one. 

Assume $P(X,Z)=\sum_{j=0}^{k-1}p_j(Z)X^{q^j}$. Since $\deg_Z a_0\le \sigma_e-1$, we define $S_0:=\{z\in S:a_0(z)\neq 0\}$ by removing at most $\sigma_e-1$ roots of $a_0(Z)$ from $S$ and $S_0(Z):=\prod_{z\in S_0}(Z-z)$. Reducing Equation~\eqref{eq:compositionmodS} modulo $S_0(Z)$ and comparing the coefficient of $X^{q^j}$ for $0\le j\le k-1$, we have a triangular recurrence
\[
a_0(Z)p_j(Z) + \sum_{r=1}^{\min\{e,j\}} a_r(Z)p_{j-r}(Z)^{q^r} + b_j(Z) =0 \pmod{ S_0(Z)}.
\]
Since $a_0(Z)$ is invertible modulo $S_0(Z)$, the coefficients $p_j(Z)$ can be recursively reconstructed. We clear all denominators by multiplying $a_0(Z)^{\sigma_{k-1}}$ in this recurrence and define \[\tilde P(X,Z):=a_0(Z)^{\sigma_{k-1}}P(X,Z) \pmod{S_0(Z)}.\] 
By Moore interpolation (see Lemma~\ref{lem:int}), let $\widehat{\bu}_s(X)$ be the $q$-linearized polynomial interpolated from $\bu_s$ over the domain $\cA$ for $s\in\{0,1\}$, and define, for $x\in V:=\mathrm{span}_{\mathbb F_q}(\mathcal{A})$,
\[
  H_x(Z):=\widetilde P(x,Z) -a_0(Z)^{\sigma_{k-1}}\bigl(\widehat{\bu}_0(x)+Z\widehat{\bu}_1(x)\bigr) \pmod{S_0(Z)}.
\]
By Lemma~\ref{lem:bw-interpolation}, we have $\deg_Z H_x\le (\sigma_e-1)\sigma_{k-1}+1$.  For each $z\in S_0$, since $P(X,z)-(\widehat{\bu}_0(x)+z\widehat{\bu}_1(x))$ vanishes over a subspace of $V$ of $\F_q$-dimension at least $n-e$, the kernel of the $\mathbb F_q$-linear map $x\mapsto H_x(z)$ also has dimension at least $n-e$. Hence, it intersects with every $(e+1)$-dimensional subspace $W$ of $V$ nontrivially, i.e., $\exists\ 0\neq x\in W$ such that $H_x(z)=0$. Moreover, $H_{cx}(z)=0$ for any $c\in\F_q$ since ${P}(X,z)$, $\hat{u}_0(X)$, and $\hat{u}_1(X)$ are all $\F_q$-linear. Let $\{\F_q x_j: j=1,\dots,\sigma_e\}$ be all one-dimensional subspaces of $W$. Considering the product $\prod_{j=1}^{\sigma_e}H_{x_j}(Z)$, for every $z\in S_0$, at least one factor $H_{x_j}(z)$ is zero, hence it vanishes on $S_0$. Since $|S_0|>\tau_e-(\sigma_e-1)=\sigma_e((\sigma_e-1)\sigma_{k-1}+1)$, root counting forces one factor to vanish identically. Thus, the set $E^*:=\{x\in V:H_x(Z)=0\}$ meets every $(e+1)$-dimensional subspace of $V$. As $E^*$ is an $\F_q$-subspace, this gives $\dim_{\mathbb F_q}E^*\ge n-e$.

Choose an $(n-e)$-dimensional subspace $E\subseteq E^*$. On the one hand, for all $x\in E$, $\widetilde P(x,Z)=a_0(Z)^{\sigma_{k-1}}\bigl(\widehat{\bu}_0(x)+Z\widehat{\bu}_1(x)\bigr)$. On the other hand, since $n-e\ge k$, Moore interpolation produces two $q$-linearized polynomials $P_0,P_1$ of degree $<k$ that agree with $\widehat{\bu}_0,\widehat{\bu}_1$ on $k$ independent points of $E$. The defining identity of $E$ then implies, as an identity of $q$-linearized polynomials over $\mathbb F_{q^m}(Z)$,
\[
  \widetilde P(X,Z)
  =a_0(Z)^{\sigma_{k-1}}\bigl(P_0(X)+ZP_1(X)\bigr).
\]
Canceling the nonzero factor $a_0(Z)^{\sigma_{k-1}}$ and comparing the coefficients of $1$ and $Z$ shows that $P_s$ agrees with $\widehat{\bu}_s$ on the subspace $E$ for $s=0,1$. Setting $\bc_s=P_s(\mathcal{A})$ gives the desired correlated rank bound.

\medskip
\noindent\textbf{Counterexample at proximity $\delta=\frac{d}{2n}$.} We outline the construction. Fix $g\ge 1$, set $Q=q^g$, and let $F=\mathbb F_{Q^6}=\mathbb F_{q^{6g}}, n=6g.$ Let $\mathcal A=\{\alpha_1,\ldots,\alpha_n\}$ be an $\mathbb F_q$-basis of $F$, and take $\cC=\mathcal G(\mathcal A,2g+1)\subseteq F^n.$ Then $\cC$ has minimum distance $d=n-(2g+1)+1=4g.$ 

Define $\widehat{\bu}_0(X)=X^{Q^4}, \widehat{\bu}_1(X)=X^{Q^3},$
and let $\bu_s=\widehat{\bu}_s(\mathcal A)$ for $s\in\{0,1\}$. We show that for every $\lambda\in F$, the word $\bu_0+\lambda \bu_1$ is $\frac{d}{2n}$-close to $\cC$, while $\bu_1$ itself is at distance at least $\frac{3d}{4n}$-far from $\cC$.

For a four-dimensional $\mathbb F_Q$-subspace $W\subseteq F$, write its annihilator polynomial as
\[
A_W(X):=\prod_{w\in W}(X-w)=X^{Q^4}+\Lambda(W)X^{Q^3}+c_1X^{Q^2}+c_2X^Q+c_3X.
\]
Let $\lambda=\Lambda(W)$ be the second-highest coefficient and $P_\lambda(X)=-c_1X^{Q^2}-c_2X^Q-c_3X$. Then we have $P_\lambda(\mathcal A)\in \cC$ and
\[
\widehat{\bu}_0(X)+\lambda\widehat{\bu}_1(X)-P_\lambda(X)=A_W(X).
\]
Since $\ker A_W=W$ has $\mathbb F_q$-dimension $4g$, the image of $A_W:F\to F$ has $\mathbb F_q$-dimension $2g=d/2$. Thus $\bu_0+\lambda\bu_1$ is $\delta$-close to $\cC$. It remains to show that the second-highest coefficient $\Lambda(W)$, as $W$ ranges over all four-dimensional $\mathbb F_Q$-subspaces of $F$, ranges over all of $F$.

This counting problem reduces to dimension two. Let $U\subseteq F$ be a two-dimensional $\mathbb F_Q$-subspace. Assume $A_U(X)=X^{Q^2}+\alpha X^Q+\mu X$ and set $W=A_U(F)$. Then $W$ is four-dimensional over $\mathbb F_Q$, and
\[
A_W(A_U(X))=X^{Q^6}-X.
\]
Comparing both coefficients of $X^{Q^5}$ gives $\Lambda(W)=-\alpha^{Q^4}$. Hence the second highest coefficient $\Lambda(W)$ of a four-dimensional vector space $W$ is determined by the second-highest coefficient of the corresponding two-dimensional vector space $U$. By Lemma~\ref{lem:lift-to-four}, it suffices to count the size of $L_2=\{\Lambda(U):\dim_{\mathbb F_Q}U=2\}.$

For this count, write $U_{x,r}=\operatorname{span}_{\mathbb F_Q}\{x,xr\},$
where $x\in F^*=F\setminus\{0\}$ and $r\in F\setminus \mathbb F_{Q^2}$. Since both $x$ and $xr$ are roots of $A_{U_{x,r}}$, one obtains
\[
\Lambda(U_{x,r})=-x^{Q^2-Q}R(r),
\quad
R(r)=\frac{r^{Q^2}-r}{r^Q-r}.
\]
For fixed $r$, as $x$ ranges over $F^*$, the factor $x^{Q^2-Q}$ ranges exactly over the kernel of norm $\mathrm N_{F/\mathbb F_Q}$. Therefore the possible values of $\Lambda(U_{x,r})$ form a single norm coset:
\[
\{\Lambda(U_{x,r}):x\in F^*\}
=
\{a\in F^*:\mathrm N_{F/\mathbb F_Q}(a)
=
\mathrm N_{F/\mathbb F_Q}(-R(r))\}.
\]
Thus it remains to show that the norm values $\mathrm N_{F/\mathbb F_Q}(-R(r))$
cover all of $\mathbb F_Q^*$. Lemma~\ref{lem:two-dim-diversity} proves this using Weil bounds for the multiplicative character sums associated with $R(X)=\frac{X^{Q^2}-X}{X^Q-X},$ together with character orthogonality. Keeping the explicit constants gives a positive lower bound on the number of occurrences of each nonzero norm value. The coefficient zero is obtained by taking $U=\mathbb F_{Q^2}$. 

Consequently, $L_2=F$. By Lemma~\ref{lem:lift-to-four}, we also have $L_4=\{\Lambda(W):\dim_{\mathbb F_Q}W=4\}=F.$
Hence every $\lambda\in F$ arises as $\Lambda(W)$ for some four-dimensional $W$, and gives a point $\bu_0+\lambda\bu_1$ that is $\delta$-close to $\cC$.

Finally, this affine line has no correlated-agreement explanation. Looking only at the second row, for every $q$-linearized polynomial $P$ of degree less than $k=2g+1$, 
the polynomial $X^{Q^3}-P(X)$ is a nonzero $q$-linearized polynomial of $q$-degree at most $3g$. Its image on $F$ has $\mathbb F_q$-dimension at least $3g$, and therefore $d(\bu_1,\cC)\ge 3g=\frac{3d}{4}$ and $d([\bu_0,\bu_1],\cC^2)\ge \frac{3d}{4}$.

\medskip
\noindent\textbf{A tight soundness lower bound at $\delta=\frac{d-1}{2n}$.} We use the parameters of Theorem~\ref{thm:main-half-threshold-counterexample} and choose the evaluation points to be an $\F_q$-basis of a $c$-dimensional $\F_Q$-subspace $V\subseteq\F_{q^m}$, where $Q=q^{e+1}$. For each $\F_Q$-hyperplane $W\subseteq V$, the image $A_W(V)$ is one-dimensional over $\F_Q$. For $b\in A_W(V)\setminus\{0\}$, set $\lambda=-b^{q-1}$. The polynomial $A_W(X)^q+\lambda A_W(X)$ has rank $e$ on $V$, since the map $y\mapsto y^q+\lambda y$ has kernel $b\F_q$ on $A_W(V)=b\F_Q$. Its two highest terms are $X^{q^{n-e}}+\lambda X^{q^{n-e-1}}$ and all remaining terms have $q$-degree less than $k=n-2e$. Thus each such $\lambda$ gives an $e$-close point on a fixed affine line. Lemma~\ref{lem:hyperplane-image-parameters} counts exactly $(q^n-1)/(q-1)$ distinct such parameters. The direction vector $\bu_1$ remains at distance at least $e+1$ from the code by the $q$-degree bound. This gives the soundness lower bound in Theorem~\ref{thm:half-distance-threshold-counterexample}.

\medskip
\noindent\textbf{The construction of rank metric $\LRM$-PCS.} The central ingredient of the original Ligero \cite{ligero} is a tensor-product decomposition of univariate polynomials. To realize such decomposition for a $q$-linearized polynomial, we lift the coefficient vector by a power of $q$. Assume $k = \kappa^2$ is a square satisfying $\kappa \mid m$. Let $f(x) = \sum_{\ell=0}^{k-1} f_\ell x^{q^\ell}$ be a $q$-linearized polynomial over $\mathbb{F}_{q^{m}}$ of $q$-degree less than $k$. After lifting each coefficient $f_{\ell}$ to $f^{q^{m-{\lfloor\ell/\kappa \rfloor}\kappa}}_{\ell}$, the resulting vector is reshaped into a $\kappa \times \kappa$ matrix $\mF$. In this view, $f(x)$ can be expressed as a tensor product under symbolic product:
\begin{equation}\label{eq:tensordecom}
    f(x) =\left(x, x^{q^{\kappa}}, \dots, x^{q^{\kappa(\kappa-1)}}\right) \odot \left(\, {\mF} \cdot \left(x, x^q, \dots, x^{q^{\kappa-1}}\right)^{T}\right),
\end{equation}
where ``$\cdot$'' denotes the standard matrix vector multiplication, and ``$\odot$'' denotes an inner symbolic product defined before Lemma~\ref{lem:comp}.

To commit, we first encode the columns of $\mF$ using an efficient Gabidulin encoder with quasi-linear-time, which is obtained by choosing the evaluation points $\{\alpha_1,\dots,\alpha_n\}$ to be a normal basis of $\mathbb F_{q^n}$~\cite{PW18}. This produces a matrix $U \in \mathbb{F}_{q^m}^{n \times \kappa}$, where each column $U[j]$ is a Gabidulin codeword. The resulting scheme takes the form of an IOPP for interleaved Gabidulin codes. Specifically, the prover commits to $U$ by applying a Merkle tree hash to each of its columns. The verifier then queries a small subset of matrix rows and checks their local consistency. In the evaluation phase, for any $\alpha\in \F_{q^m}$, from the Equation~\eqref{eq:tensordecom}, we have
\[  f(\alpha) =\left(x, x^{q^{\kappa}}, \dots, x^{q^{\kappa(\kappa-1)}}\right) \odot \left(\, {\mF} \cdot \left(\alpha, \alpha^q, \dots, \alpha^{q^{\kappa-1}}\right)^{T}\right),\]
after using same proximity test for the correctness of $\bv={\mF} \cdot \left(\alpha, \alpha^q, \dots, \alpha^{q^{\kappa-1}}\right)^{T}$, the verifier then computes and obtains $f(\alpha)$. The soundness of $\LRM$-PCS rests on two components: the soundness of the proximity test of the Gabidulin codes and the collision resistance of the Merkle hash.

\paragraph{Acknowledgment of AI assistance.}
Apart from limited AI assistance for language editing, no AI assistance was used in obtaining the results of this paper.

%% file: sections/preliminary.tex
\section{Preliminaries}\label{sec:2}
In this section, we briefly introduce some definitions and basic results about linearized polynomials, Gabidulin codes and polynomial commitments.

Throughout this paper, $q$ is always a power of a prime and $\F_q$ denotes the finite field of $q$ elements. Let $\F_{q^m}$ be an extension field of $\F_q$ of degree $m$. For an integer $n\geq 1$, we denote by $[n]$ the set $\{1,2,\dots,n\}$. 

\subsection{Linearized Polynomials}\label{sec:linearizedpoly}
Let $1 \leq k \leq m$. A $q$-linearized polynomial over $\mathbb{F}_{q^m}$ is a polynomial of the form $f(x) = \sum_{i=0}^{k-1} f_i x^{q^i}\in \mathbb{F}_{q^m}[x]$. If the highest coefficient $f_{k-1}\neq 0$, we say that the $q$-degree of $f(x)$ is $k-1$, denoted by $\deg_q f(x)$. For the zero polynomial, we set $\deg_q(0)=-\infty$.

A $q$-linearized polynomial has the property of being $\mathbb{F}_q$-linear functions on $\mathbb{F}_{q^m}$. That is, for all $a,b \in \mathbb{F}_q$ and $\alpha,\beta \in \mathbb{F}_{q^m}$, one has
\[
f(a\alpha + b\beta) = af(\alpha) + bf(\beta).
\]
We denote by $\mathcal{L}(q,m)$ the set of $q$-linearized polynomials over $\F_{q^m}$, i.e.,
\[\mathcal{L}(q,m):=\left\{f(x)=\sum_{i=0}^{\ell}f_ix^{q^i}:f_i\in\F_{q^m},\ell\geq 0\right\}.\]
Furthermore, denote by $\mathcal{L}(q,m,k)$ the set of $q$-linearized polynomials of $q$-degree at most $k-1$ over $\F_{q^m}$, i.e.,
\[\mathcal{L}(q,m,k):=\{f(x)\in\mathcal{L}(q,m):\deg_q f(x)< k\}.\]
It is obvious that $\mathcal{L}(q,m,k)$ is an $\F_{q^m}$-vector space with dimension $k$. 

For any two $q$-linearized polynomials $f(x)=\sum_{i=0}^{\ell_1}f_ix^{q^i}$ and $h(x)=\sum_{j=0}^{\ell_2}h_jx^{q^j}$, their symbolic product is defined as the composition
\[f(x)\circ h(x):=f(h(x))=\sum_{s=0}^{\ell_1+\ell_2}\sum_{i+j=s}f_ih_j^{q^i}x^{q^s}.\] 
It is clear that if $\deg_q f(x)=\ell_1$ and $\deg_q h(x)=\ell_2$, then $\deg_q (f(x)\circ h(x))=\deg_q(f(h(x)))=\ell_1+\ell_2$. The symbolic product is associative and distributive with respect to usual polynomial addition, but is non-commutative, i.e.,
\[f\circ(g+h)=f\circ g+f\circ h; \quad (g+h)\circ f=g\circ f+h\circ f\]
and
\[f\circ(g\circ h)=(f\circ g)\circ h\]
for any $f,g,h\in\mL(q,m)$. Therefore, $\mL(q,m)$ is a non-commutative ring associated with polynomial addition and symbolic product.

For ordinary polynomials, we have the long division. In the case of linearized polynomials, we have also ``long division'' by replacing the usual polynomial product with symbolic product. 
\begin{lemma}[\cite{ore1933theory}]~\label{lem:2.1}
Let $f(x), g(x)\in \mathcal{L}(q,m)$ with $g\neq 0.$ Then, there exists a unique pair $(h,r)$ with $h, r\in \mathcal{L}(q,m)$ and $\deg_q(r)<\deg_q(g)$ satisfying
\begin{equation}\label{eq:1_}
    f=h\circ g+r.
\end{equation}
In particular, if $r=0$, then we say that $g$ right divides $f$.
\end{lemma}

The following lemma presents interpolation of $q$-linearized polynomials over $\F_{q^m}$.
\begin{lemma}[Interpolation]\label{lem:int}
Assume $\Ga_1,\ldots,\Ga_n\in \mathbb F_{q^m}$ are $\F_q$-linearly independent. Then, for any $y_1,\ldots,y_n\in \mathbb F_{q^m}$, there exists a unique polynomial $f\in \mL(q,m,n)$ such that $ f(\Ga_j)=y_j,\ j=1,\ldots,n.$
\end{lemma}
\begin{proof}
The interpolation conditions are equivalent to the linear system
\[
\begin{pmatrix}
\Ga_1 & \Ga_1^q & \cdots & \Ga_1^{q^{n-1}} \\
\Ga_2 & \Ga_2^q & \cdots & \Ga_2^{q^{n-1}} \\
\vdots & \vdots & & \vdots \\
\Ga_n & \Ga_n^q & \cdots & \Ga_n^{q^{n-1}}
\end{pmatrix}
\begin{pmatrix}
a_0\\
a_1\\
\vdots\\
a_{n-1}
\end{pmatrix}
=
\begin{pmatrix}
y_1\\
y_2\\
\vdots\\
y_n
\end{pmatrix}.
\]
The coefficient matrix is a Moore matrix, which is nonsingular whenever
$\Ga_1,\ldots,\Ga_n$ are linearly independent over $\mathbb F_q$. Hence the system has a unique solution $(a_0,\ldots,a_{n-1})$, giving the desired unique interpolating $q$-linearized polynomial.
\end{proof}

\subsection{Rank-Metric Codes and Gabidulin Codes}
 Let $\F_q^{m\times n}$ denote the set of all $m\times n$ matrices over $\F_q$. For any matrix $M\in \F_q^{m\times n}$, let $\mathrm{rk}_q(M)$ denote its rank over $\F_q$. If we fix an $\F_q$-basis of $\F_{q^m}$, then an element of $\F_{q^m}$ can be identified with a column vector in $\F_q^m$ and hence a vector $\bc=(c_1,c_2,\dots,c_n)\in\F_{q^m}^n$ can be viewed as a matrix in $\F_q^{m\times n}$. For any two vectors $\bu,\bv\in\F_{q^m}^n$, their rank distance and relative rank distance are defined as
 \[d_{\rk}(\bu,\bv)=\mathrm{rk}_q(\bu-\bv),\ \Delta_{\rk}(\bu,\bv)=d(\bu,\bv)/n.\]
In the following, unless otherwise specified, we simply use $d(\cdot,\cdot)$ and $\Delta(\cdot,\cdot)$ to denote the rank distance $d_{\rk}(\cdot,\cdot)$ and relative rank distance $\Delta_{\rk}(\cdot,\cdot)$, respectively,

Assume $k\leq n\leq m$. A $[n,k]$ linear rank-metric code $\mC$ over $\F_{q^m}$ is an $\F_{q^m}$-subspace of $\F_{q^m}^n$ of dimension $k$. Its rank distance $d(\mC)$ is defined by
 \[d(\mC)=\min\{d(\bc_1,\bc_2):{\bc_1}\neq\bc_2\in \mC\}.\]
If $d(\mC)=d$, then $\mC$ is called an $[n,k,d]$ linear rank-metric code.

\begin{definition}
Let $\mC\subseteq \F_{q^m}^n$ be a rank-metric code. For any word $\bv\in\F_{q^m}^n$, its rank distance and relative rank distance to $\cC$ is defined as
\[d(\bv,\cC)=\min\{\rk_q(\bv-\bc)\mid \bc\in \cC\},\ \Delta(\bv,\cC):=d(\bv,\cC)/n.\]
If $\Delta(\bv,\cC)\le \delta$, we say $\bv$ is $\delta$-close to $\cC$. Otherwise, we say $\bv$ is $\delta$-far from $\cC$.
\end{definition}

An important class of rank-metric codes is the Gabidulin code~\cite{G85}, which is widely regarded as the rank-metric analogues of Reed--Solomon codes. 
\begin{definition}\label{def:Gabidulin-codes}
    Let $\cA=\{\Ga_1, \Ga_2,\cdots, \Ga_n\}\subset\F_{q^m}$ be an $\F_q$-linearly independent set. For $k\leq n\leq m$,
The Gabidulin code of length $n$ and dimension $k$ associated with the evaluation set $\cA$ is defined as
\[\mathcal{G}(\cA,k):=\{(f(\Ga_1),f(\Ga_2),\cdots,f(\Ga_n)): f(x)\in\mL(q,m,k)\}.\]
It has rank distance $d=n-k+1$.
\end{definition}

\begin{definition}{(Interleaved rank-metric code \cite{BartzHLPRW22}}\label{def:Interleaved-code}
    Let $\cC$ be an $[n,k,d]$ $\F_{q^m}$-linear rank-metric code over $\F_{q^m}$ and $s$ be a positive integer. We let $\cC^s$ denote the interleaved code over $(\F_{q^m})^s$ whose codewords are all $s\times n$ matrices $U$ such that every row of $U$ is a codeword in $\cC$. For a matrix $U$, we denote its $i$th row and $j$th column by $U(i)$ and $U[j]$, respectively. The rank of such a matrix is always taken after viewing $U$ as an $ms\times n$ matrix over $\F_q$. 
\end{definition}
In particular, the interleaved Gabidulin code $\cG(\mathcal{A},k)^s$ attains the same minimum distance $d = n - k + 1$ as the underlying Gabidulin code $\cG(\mathcal{A},k)$.


\subsection{The Berlekamp--Welch Unique Decoding Algorithm for Gabidulin codes}\label{sec:BW}
We recall the Berlekamp--Welch decoding method for Gabidulin codes introduced in~\cite{Loidreau2006Welch}. The main idea is to reduce unique decoding to constructing a polynomial in $\mL(q,m,k)$ by solving a linear system.

In the following, let $\cC=\mathcal G(\mathcal{A},k)\subseteq\mathbb F_{q^m}^n$ be a Gabidulin code. For a $q$-linearized polynomial $P$, we write the evaluation of $P$ over the set $\cA$ by 
$$P(\mathcal{A}):=(P(\alpha_1),\ldots,P(\alpha_n)).$$ We also fix a decoding radius $t$ satisfying $2t<d=n-k+1.$

\begin{definition}\label{def:linearized-reconstruction}
Given a vector $\by\in F_{q^m}^n$ such that $d(\by,\cC)\leq t$, the linearized polynomial reconstruction problem associated with $(\by,t;\cA,k)$ is to find a nonzero pair of $q$-linearized polynomials
\[
\Lambda(X)=\sum_{r=0}^{t}\lambda_rX^{q^r},\quad
N(X)=\sum_{j=0}^{k+t-1}\nu_jX^{q^j},
\]
such that
\begin{equation}\label{eq:bw-reconstruction}
\Lambda(y_i)=N(\alpha_i),\quad i\in[n].
\end{equation}
\end{definition}

The reconstruction system is homogeneous. Hence the pair $(\Lambda,N)$ is not meant to be unique: if $(\Lambda,N)$ is a solution, then so is $(c\Lambda,cN)$ for every $c\in\mathbb F_{q^m}^{*}$. This nonuniqueness does not affect decoding, because the decoder relies only on the compositional relation between $\Lambda$ and $N$. Specifically, when the received word is close to a unique codeword $P(\mathcal{A})$, the desired solution satisfies $N=\Lambda\circ P.$ Thus, the reconstruction problem can be viewed as implicitly encoding the unknown message polynomial $P$ in the composition $N=\Lambda\circ P$.

The following standard result explains why such a reconstruction pair exists whenever the received word is within the unique-decoding radius.

\begin{lemma}[\cite{Loidreau2006Welch}]\label{lem:bw-existence}
Suppose that $\by=P(\mathcal{A})+\mathbf{e}$ for some $P\in\mathcal L(q,m,k)$ and $\be\in\F_{q^m}^n$ with $\mathrm{rk}_q(\mathbf{e})\le t.$ Then the reconstruction Equation \eqref{eq:bw-reconstruction} is solvable. More specifically, one may take $\Lambda$ to be the subspace annihilator polynomial of the error space Span$_q\langle e_1,\dots,e_n\rangle$ and $N=\Lambda\circ P$. Moreover, any nonzero solution $(A, B)$ to Equation~\eqref{eq:bw-reconstruction} satisfying $\deg_q A \le t$ and $\deg_q B \le k + t - 1$ must satisfy the composition identity $B = A \circ P$. Consequently, the message polynomial $P$ can be uniquely recovered from $A$ and $B$.
\end{lemma}

Therefore, the Berlekamp--Welch decoder can be summarized as follows. First, solve the homogeneous reconstruction system to obtain a nonzero pair $(\Lambda,N)$. Second, recover the message polynomial $P$ from the composition identity $N=\Lambda\circ P$. Finally, verify that the recovered codeword is indeed within rank distance $t$ of the received word.

\subsection{Interactive Oracle Proofs} \label{sec:iop}
We briefly recall the definitions of interactive oracle proofs (IOPs) and interactive oracle proofs of proximity for error-correction codes \cite{basefold}.
\paragraph{Interactive oracle proofs (IOPs).} An $r$-round public coin $\IOP=(\Po, \V)$, for a relation $\mathcal{R}$ runs as follows: Initially, $\Po$ sends an oracle string $\pi_0$. In each round $i\in [1, r]$, the verifier samples and sends a random challenge $z_i$, and the prover replies with an oracle string $\pi_i$. After $r$ rounds of communications, the verifier $\V$ queries some entries of the oracle strings $\pi_0, \pi_1,\dots, \pi_r$ and outputs a bit $b$.  
\begin{definition}
   Let $\IOP=(\Po,\V)$ be a $r$-round public coin IOP protocol for a relation $\cR$. We say that $\IOP$ is complete if for every $(x, w)\in\cR$,  
   \[\Pr_{z_1,\dots,z_r} \left[\V^{\pi_0,\pi_1,\dots,\pi_r}(x,z_1,\dots,z_r)=1 \,\middle|\, \begin{array}{l}
        \pi_0 \leftarrow \Po(x,w)\\
   \pi_1 \leftarrow \Po(x,w,z_1)\\
   ~~\cdots \\
   \pi_r\leftarrow \Po(x,w,z_1,\dots,z_{r})  
   \end{array}\right]=1\]
We say that $\IOP$ has soundness error $\epsilon$ if for any $x\notin L(\cR)$ and any adversary $\cA$,
   \[\Pr_{z_1,\dots,z_r} \left[\V^{\pi_0,\pi_1,\dots,\pi_r}(x,z_1,\dots,z_r)=1 \,\,\middle|\, \begin{array}{l}
        \pi_0 \leftarrow \cA(x)\\
   \pi_1 \leftarrow \cA(x,z_1)\\
  ~~ \cdots \\
   \pi_r\leftarrow \cA(x,z_1,\dots,z_{r})  
   \end{array}\right]<\epsilon\]
\end{definition}

\paragraph{IOPs of proximity.} An IOP of proximity (IOPP) is similar to IOP with the specialty that the witness $w$ is also sent as an oracle string. The verifier can query $w$ as an oracle but will only query $q_w\ll |w|$ entries of $w$. The soundness states that if $w$ is far from any valid witness, then the verifier rejects with high probability.
\begin{definition}
A public coin $\IOPP=(\Po, \V)$ for relation $\cR$ is an IOP of proximity if it satisfies $\IOP$ completeness, and the $\epsilon(\cdot)$-IOPP soundness holds for some function $\epsilon(\cdot)$: for every $(x, w)$ where $w$ is $\delta$-far (in relative Rank distance) from any $w'$ such that $(x, w')\in \cR$, it holds that for any adversary ${\cA}$,
   \[\Pr_{z_1,\dots,z_r} \left[\V^{w,\pi_0,\dots,\pi_r}(x,z_1,\dots,z_r)=1 \,\,\middle|\, \begin{array}{l}
        \pi_0 \leftarrow \cA(x,w)\\
   \pi_1 \leftarrow \cA(x,w,z_1)\\
   ~~\cdots \\
   \pi_r\leftarrow \cA(x,w,z_1,\dots,z_{r})  
   \end{array}\right]<\epsilon(\delta)\]
\end{definition}

\subsection{Polynomial Commitment Schemes}\label{sec:pcs}
We adapt the definition from~\cite{BFS20,basefold} to define a rank-metric polynomial commitment scheme for $q$-linearized polynomials.

\begin{definition}\label{def:pc}
Let $\mathcal L(q,m,k)$ be the space of $q$-linearized polynomials of $q$-degree less than $k$. A tuple of algorithms $(\mathsf{Gen},\mathsf{Commit},\mathsf{Open},\mathsf{Eval})$ is an \emph{extractable polynomial commitment scheme} for $\mathcal L(q,m,k)$ if the following properties hold.
\begin{description}
\item[--] $\mathsf{Gen}(1^\lambda,q,m,k)$ takes security parameter $\lambda$, prime power $q$, positive integers $k\leq m$ and outputs public parameters $\mathsf{pp}$.
\item[--] $\mathsf{Commit}(\mathsf{pp},f)$ takes a $q$-linearized polynomial $f\in\mathcal L(q,m,k)$ and outputs a commitment $\mathsf{com}$. 
\item[--] $\mathsf{Open}(\mathsf{pp},f,\mathsf{com})$ takes a $q$-linearized polynomial $f$ and a commitment $\mathsf{com}$, and outputs a bit $b\in\{0,1\}$. 
\item[--] $\mathsf{Eval}(\mathsf{pp},\mathsf{com},\alpha,\beta)$ is an interactive public-coin protocol between a probabilistic polynomial time (PPT) prover $\Po$ and verifier $\V$ both with input $\mathsf{pp}, \mathsf{com},$ an evaluation point $\alpha\in\F_{q^m}$ and a claimed evaluation $\beta\in\F_{q^m}$. $\Po$ additionally knows the opening of $\mathsf{com}$ to a $q$-linearized polynomial $f\in\mL(q,m,k)$. The protocol convinces the verifier that $\beta=f(\alpha)$. At the end of the protocol, $\V$ outputs $b\in\{0,1\}$. 
\item[Completeness.] An honest prover can successfully convince the verifier of any evaluation. Specifically, for every \(f\in\mathcal L(q,m,k)\) and every $\alpha\in\mathbb F_{q^m}$, if $\mathsf{com}\leftarrow \mathsf{Commit}(\mathsf{pp},f)$ and $\beta=f(\alpha)$, then 
\[\mathsf{Eval}(\mathsf{pp},\mathsf{com},\alpha,\beta)=1.\]
\item[Binding.] The PCS is binding if for every PPT adversary $\mathcal A$, 
\[\Pr\left[
   \begin{array}{l}
   \mathsf{pp}\leftarrow \mathsf{Gen}(1^\lambda,q,m,k);\ (\mathsf{com},f_0,f_1)
   \leftarrow \mathcal A(\mathsf{pp});f_0\neq f_1\\
   \mathsf{Open}(\mathsf{pp},f_0,\mathsf{com})=1;\ \mathsf{Open}(\mathsf{pp},f_1,\mathsf{com})=1;\ 
   \end{array}
\right]
\leq \mathrm{negl}(\lambda).
\]
\item[Knowledge Soundness.]  {Any successful prover in the $\mathsf{Eval}$ protocol must know a polynomial $f\in\mL(q,m,k)$ such that $f(\alpha)=\beta$ and $\mathsf{com}$ is a commitment to $f(x)$. More specifically, $\mathsf{Eval}$ is a succinct argument of knowledge for the relation given $\mathsf{pp}\leftarrow \mathsf{Gen}(1^{\lambda},q,m,k)$
\[
\mathcal{R}_{\mathsf{Eval}}(\mathsf{pp})
:=
\left\{
\left(
(\mathsf{com},\alpha,\beta),f
\right)
:
\begin{array}{l}
f\in\mathcal{L}(q,m,k),\\
\mathsf{Open}(\mathsf{pp},\mathsf{com},f)=1,\\
f(\alpha)=\beta
\end{array}
\right\}.
\]
That is for any PPT algorithms $\cA$ and $\Po^*$, there exists a PPT extractor $\mathsf{Ext}$ such that,
for any randomness $\rho$,
\[
\left|\Pr\!\left[
\begin{array}{l}
\mathsf{pp}\leftarrow\mathsf{Gen}(1^\lambda,q,m,k),\\
(\mathsf{com},\alpha,\beta)\leftarrow\mathcal{A}(\mathsf{pp};\rho),\\
\mathsf{Eval}(\mathsf{pp},((\mathsf{com},\alpha,\beta)))
\langle\mathcal{P}^*(\rho),\V\rangle=1
\end{array}
\right]-
\Pr\!\left[
\begin{array}{l}
\mathsf{pp}\leftarrow\mathsf{Gen}(1^\lambda,q,m,k),\\
(\mathsf{com},\alpha,\beta)\leftarrow\mathcal{A}(\mathsf{pp};\rho),\\
f\leftarrow\mathsf{Ext}(\mathsf{pp},(\mathsf{com},\rho))\\
((\mathsf{com},\alpha,\beta),f)\in \mathcal{R}_{\mathsf{Eval}}(\mathsf{pp})
\end{array}
\right]\right|
\leq\operatorname{negl}(\lambda).
\]}
\end{description}
\end{definition}

%% file: sections/proximity-gap.tex
\section{Rank-Metric Proximity Gaps}
In this section, we prove proximity gaps for linear rank-metric codes, which will serve as a key ingredient in our $q$-linearized PCS.

\subsection{A $\frac{d-1}{3n}$-Proximity Gap for General Linear Rank-Metric Codes}

\begin{lemma}\label{lem:d/3gap}
    Let $\cC\subseteq \F_{q^m}^n$ be an $\F_{q^m}$-linear rank metric code with minimum distance $d$. For any two vectors $\bu, \bv$ in $\F_{q^m}^n$, let $\ell_{\bu,\bv}:=\{\bu+\lambda \bv\mid \lambda\in\F_{q^m}\}$. Let $1\le e \le \lfloor (d-1)/3\rfloor$ be a positive integer and $q^{e+1}<q^m$. Then either all points $\bx$ in $\ell_{\bu,\bv}$ satisfy the rank distance $d(\bx,\cC)\leq e$ or 
    \[\Pr_{\lambda\in\F_{q^m}}[d(\bu+\lambda \bv,\cC)\leq e]\leq q^{e+1}/q^m\]
\end{lemma}
\begin{proof}
    If $\bv = 0$, it is clearly true. Now assume that $\bv \neq 0$. If $\bu,\bv$ are $\F_{q^m}$-linearly dependent, then $\bu+\lambda \bv=(\alpha+ \lambda)\bv$ for some $\alpha\in\F_{q^m}$. Thus it is also clearly true.
Now we assume that $\bu, \bv$ are $\F_{q^m}$-linearly independent. 

Suppose there exists at least one point $\bx = \bu + \alpha \bv \in \ell_{\bu,\bv}$ such that $d(\bx,\cC) > e$. Since $\ell_{\bu,\bv} = \ell_{\bx,\bv}$, we may equivalently replace $\bu$ with $\bx$ and assume that $d(\bu,\cC) > e$. 

Next, for any nonzero $\lambda \in \mathbb{F}_{q^m}^*$, the vector $\lambda^{-1}(\bu+\lambda \bv) = \lambda^{-1}\bu + \bv$. Note that if $d(\bv,\cC) \le e$, we have the following inequality
\begin{equation}\label{eq:step1}
\begin{split}
\Pr_{\lambda \in \mathbb{F}_{q^m}} [ d(\bu+\lambda \bv , \cC) \le e ]
&=\frac{q^m-1}{q^m}\cdot\Pr_{\lambda \in \mathbb{F}_{q^m}^*} [ d(\bu+\lambda \bv , \cC) \le e ] \\[6pt]
&= \frac{q^m-1}{q^m}\cdot\Pr_{\lambda \in \mathbb{F}_{q^m}^*} [ d(\lambda^{-1}\bu + \bv , \cC) \le e ] \\[6pt]
&= \Pr_{\lambda \in \mathbb{F}_{q^m}}  [d(\bv + \lambda \bu , \cC) \le e ] - \frac{1}{q^m}.
\end{split}
\end{equation}
Thus, if $d(\bv,\cC) \le e$, then the probability that $\bu+\lambda \bv$ lies within distance $e$ of $\cC$ would be strictly smaller than the corresponding probability for $\bv+\lambda \bu$, where $d(\bu,\cC)>e$. Thus, it suffices to show the result under the assumption that $d(\bv,\cC) > e$.

If there is only one $\lambda \in \mathbb{F}_{q^m}$ such that $d(\bu+\lambda \bv, \cC) \leq e$, we are done.  

Now assume that there are at least two distinct $\lambda_1, \lambda_2 \in \mathbb{F}_{q^m}$ such that 
\[
d(\bu+\lambda_i \bv, \cC) \leq e \quad \text{for } i=1,2.
\] 
For $i=1,2$, let $\bc_i \in \cC$ such that $d(\bu+\lambda_i \bv, \bc_i) \le e$, put
\[
\bu' = \bu+\lambda_1 \bv - \bc_1, \qquad \bv' = \bu+\lambda_2 \bv - \bc_2.
\]
Then $\rk_q(\bu'),\rk_q(\bv') \le e$. 

For any $\alpha\in\F_{q^m}$,
\[
\bu' + \alpha \bv' 
= (1+\alpha) \bu + (\lambda_1+\alpha \lambda_2)\bv - (\bc_1+\alpha \bc_2),
\]
If $\alpha \neq -1$, then
\[
d(\bu'+\alpha \bv', \cC) = d\left( \bu + \frac{\lambda_1+\alpha \lambda_2}{1+\alpha} \bv, \cC \right).
\]
Define a set $S$ as
\[
S := \left\{ \frac{\lambda_1 + \alpha \lambda_2}{1+\alpha} : \alpha \in \mathbb{F}_{q^m} \setminus \{-1\} \right\}.
\]
As $\frac{\lambda_1 + \alpha_1 \lambda_2}{1+\alpha_1} \neq \frac{\lambda_1 + \alpha_2 \lambda_2}{1+\alpha_2}$ for $\alpha_1 \neq \alpha_2$, we have $|S| = q^m - 1$. Thus,
\begin{equation}\label{eq:step2}
\begin{split}
    \Pr_{\lambda\in \mathbb{F}_{q^m}} [d(\bu+\lambda \bv, \cC) \leq e ] - \tfrac{1}{q^m}
    &\leq \frac{q^m-1}{q^m}\cdot\Pr_{\alpha \in \F_{q^m}\setminus\{-1\}} \left[ d\left(\bu + \tfrac{\lambda_1+\alpha \lambda_2}{1+\alpha} \bv, \cC \right) \leq e \right]\\[6pt]
   &\leq \Pr_{\alpha \in \mathbb{F}_{q^m}} [ d(\bu'+\alpha \bv', \cC) \leq e ].
\end{split} 
\end{equation}
Thus, it suffices to prove that there exist at most $q^{e+1}-1$ elements $\alpha\in\F_{q^m}$ such that $d(\bu'+\alpha \bv', \cC) \leq e$. This gives
\[
\Pr_{\lambda \in \mathbb{F}_{q^m}} \{ d(\bu+\lambda \bv, \cC) \leq e \}
\leq \Pr_{\alpha \in \mathbb{F}_{q^m}} \{ d(\bu'+\alpha \bv', \cC) \leq e \} + \frac{1}{q^m}
\leq \frac{q^{e+1}}{q^m}.
\] 

To prove this, we first claim that
$
\rk_q\begin{pmatrix} \bu' \\ \bv' \end{pmatrix} \geq e+1.
$
Otherwise, suppose that $\rk_q\begin{pmatrix} \bu' \\ \bv' \end{pmatrix} \leq e$. Then $d(\bu'+\lambda \bv', \cC) \le e$ for all $\lambda$.  
As $\bv =\frac{1}{\lambda_1-\lambda_2} (\bu' -\bv'+\bc_1-\bc_2)$, we would have $d(\bv,\cC) =d(\bu'-\bv', \cC)\leq e$, a contradiction.  Hence $\rk_q\begin{pmatrix} \bu' \\ \bv' \end{pmatrix} \geq e+1$.  

Let $\bu'=(u'_1,\dots,u'_n)$, $\bv'=(v'_1,\dots,v'_n)$. By the claim, we may assume that the first $e+1$ columns of $\begin{pmatrix} \bu' \\ \bv' \end{pmatrix}$ are $\mathbb{F}_q$-linearly independent. This implies that
$
\sum_{i=1}^{e+1} \lambda_i v'_i \quad \text{and} \quad \sum_{i=1}^{e+1} \lambda_i u'_i
$
cannot both be zero for $\lambda_1,\dots,\lambda_{e+1} \in \mathbb{F}_q$ not all zero. 

Note that if for some $\alpha$ such that $d(\bu'+\alpha \bv', \cC) \leq e$, then there exists $\bc \in \cC$ such that
\[
\rk_q(\bu'+\alpha \bv'-\bc) \leq e.
\]
Thus
\[
\rk_q(\bc) \le \rk_q(\bu'+\alpha \bv') + e \leq \rk_q(\bu')+\rk_q(\bv')+e \leq 3e< d,
\]
forcing $\bc=0$. Hence $\rk_q(\bu'+\alpha \bv') \le e$. Then there exist $\lambda_1,\dots,\lambda_{e+1} \in \mathbb{F}_q$ not all zero such that
\[
\sum_{i=1}^{e+1} \lambda_i (u'_i + \alpha v'_i) = 0,\ i.e.,\
\alpha \sum_{i=1}^{e+1} \lambda_i v'_i = - \sum_{i=1}^{e+1} \lambda_i u'_i.
\]
Thus $\sum_{i=1}^{e+1} \lambda_i v'_i \neq 0$, otherwise both sides vanish, a contradiction. Hence
\[
\alpha = - \frac{\sum_{i=1}^{e+1} \lambda_i u'_i}{\sum_{i=1}^{e+1} \lambda_i v'_i}.
\]
We have only $q^{e+1}-1$ choices for $(\lambda_1,\dots,\lambda_{e+1})$, and therefore at most $q^{e+1}-1$ elements $\alpha\in \F_{q^m}$ such that $d(\bu'+\alpha \bv', \cC) \leq e$. This completes the proof.

\end{proof}

Analogously to the Hamming metric setting, the above proximity gap in rank metric also implies the existence of a correlated agreement set. 

\begin{theorem}[Correlated agreement]\label{thm:ca}
Let $\cC,q,m,n,d,e$ be as defined as in Lemma~\ref{lem:d/3gap}. For $\bu_0,\bu_1\in \F_{q^m}^{n}$, if $\Pr_{z\in\F_{q^m}}[d(\bu_0+z\bu_1, \cC)\leq e]>q^{e+1}/q^m$, then
\[d\left( [\bu_0, \bu_1],\cC^2\right)\leq e.\]
\end{theorem}


\begin{proof}
    Suppose that $d\left( [\bu_0, \bu_1],\cC^2\right)> e.$ Let $U=\{\lambda_0\bu_0+\lambda_1\bu_1\}$ be the $\F_{q^m}$-linear span of $\bu_0$ and $\bu_1$. We first {claim} that
    
    \textbf{Claim}: \textit{there exists at least one $0\neq\bu^*\in U$ such that $d(\bu^*,\cC)>e$}. 

We defer the proof of this claim at the end of this proof. Then according to Lemma~\ref{lem:d/3gap}, we have $\bu^*=\beta \bu_1$ for some $\beta\in\F_{q^m}^*$ since all points in $\ell_{\bu_0,\bu_1}$ are $e$-close to $\cC$. Since 
    \[\Pr_{\lambda\in\F_{q^m}}[d(\bu_1+\lambda\bu_0,\cC)\le e]=\frac{q^m-1}{q^m}\Pr_{\lambda\in\F^*_{q^m}}[d(\lambda^{-1}\bu_1+\bu_0,\cC)\le e]=\frac{q^m-1}{q^m}\ge q^{e+1}/q^{m},\]
by Lemma~\ref{lem:d/3gap}, we also have $e\ge d(\bu_1,\cC)=d(\bu^*,\cC)$, which is a contradiction.

\textit{The proof of claim.} Assume towards contradiction that $d(\bu,\cC)\leq e$ for all $\bu\in U$. Choose $\hat{\bu}\in U$ such that $t:=d(\hat{\bu},\cC)$ is maximal. Then $t\leq e$, and we can write 
\[
\hat{\bu} = \bc + \chi,\quad\text{where}\ c\in \cC,\ \text{and}\ \chi\in\F_{q^m}^n\ \text{with}\ \rk_q(\chi)=t.
\]
Similarly, for $j=0,1$, 
\[
\bu_j = \bc_j + \chi_j,\quad\text{where}\ \bc_j\in \cC\ \text{and}\ \chi_j\in\F_{q^m}^n\ \text{with}\ \rk_q(\chi_j)\leq e.
\]
Since $d([\bu_0,\bu_1],\cC^{2}) > e$, there exists $j$ such that the $\F_q$-row span of $\chi_j$ is not contained in the $\F_q$-row span of $\chi$ when we view $\chi$ and $\chi_j$ as matrices in $\F_q^{m\times n}$. Otherwise, 
$$d([\bu_0,\bu_1],\cC^2)\leq \rk_q\left(
   [\chi_0,\chi_1]\right) \leq \rk_q(\chi)\leq e.$$
Therefore, we have \[\rk_q\begin{pmatrix}
    \chi \\ \chi_j
\end{pmatrix} \geq \rk_q(\chi)+1=t+1.\]
Assume $\chi=(x_{0,1},x_{0,2},\dots,x_{0,n})$ and $\chi_j=(x_{j,1},x_{j,2},\dots,x_{j,n})$. Without loss of generality, we may assume the first $t+1$ columns of $\begin{pmatrix}
    \chi\\ \chi_j
\end{pmatrix}$ are $\F_q$-linearly independent. This implies that
$
\sum_{i=1}^{t+1} \lambda_i x_{0,i}$ and $ \sum_{i=1}^{t+1} \lambda_i x_{j,i}
$
cannot both be zero for $\lambda_1,\dots,\lambda_{t+1} \in \mathbb{F}_q$ not all zero.  

For any $\alpha \in \mathbb{F}_{q^m}$, define $\hat{\bv}(\alpha) := \hat{\bu} + \alpha \bu_j 
= (\bc + \alpha \bc_j) + (\chi + \alpha \chi_j)$. We first claim that $\bc + \alpha \bc_j \in \cC$ is the closest codeword in $\cC$ to $\widehat{\bv}(\alpha)$. 

Indeed, let $\bc' \in \cC$ be any codeword such that $\bc'$ is closest to $\widehat{\bv}(\alpha)$. By assumption $
d(\widehat{\bv}(\alpha), \cC) \leq e$, hence we have $d(\hat{\bv}(\alpha), \bc') \leq e$. On the other hand, by construction, $d(\hat{\bv}(\alpha), \bc + \alpha \bc_j)=\rk_q(\chi + \alpha \chi_j) \leq 2e$. Suppose $\bc + \alpha \bc_j\neq \bc'$. Then, by the triangle inequality, $
d(\bc + \alpha \bc_j, \bc') \leq 2e + e <d$, this contradicts the minimum distance of $\cC$. 

Next, we show that there exists some $\alpha \in \mathbb{F}_{q^m}$ such that $
\rk_q(\chi + \alpha \chi_j) \geq t+1$. Then, we obtain $d(\widehat{\bv}(\alpha), \cC) \geq t+1,$
which contradicts the maximality of $t$.

Define a ``bad set''
\[\mathcal B := \{\alpha\in\mathbb F_{q^m}:\ \rk_q(\chi+\alpha \chi_j)\le t\}.\]
For any $\alpha\in \mathcal B$, there exist $\lambda_1,\lambda_2,\dots,\lambda_{t+1}\in \F_q$ not all zero such that \[\sum_{i=1}^{t+1}\lambda_i(x_{0,i}+\alpha x_{j,i})=0,\ \text{i.e,}\ \sum_{i=1}^{t+1}\lambda_ix_{0,i}=-\alpha\sum_{i=1}^{t+1}\lambda_ix_{j,i}.\]
Then $\sum_{i=1}^{t+1}\lambda_ix_{j,i}\neq 0$, otherwise both sides vanish, a contradiction.
Thus,
\[\alpha=-\frac{\sum_{i=1}^{t+1}\lambda_ix_{0,i}}{\sum_{i=1}^{t+1}\lambda_ix_{j,i}}.\]
Since there are at most $q^{t+1}-1\leq q^{e+1}-1$ such choices of $\alpha$, we have $|\mathcal{B}|\le q^{e+1}-1<q^m$ by assumption. Thus, there exists some $\alpha\in\F_{q^m}\setminus \mathcal B$ such that $\rk_q(\chi+\alpha\chi_j)\ge t+1$. 

\end{proof}

\subsection{A $\frac{d-1}{2n}$-Proximity Gap for Gabidulin Codes}\label{sec:d/2-gap}
Similar to the proximity gap for Reed-Solomon codes, it is possible to prove a better proximity gap for Gabidulin codes by utilizing the algebraic decoding algorithm. In this subsection, we show that the proximity gaps for Gabidulin codes can attain the unique decoding radius. 

Let $\cC=\mathcal{G}(\mathcal{A},k)$ be the Gabidulin code in Definition~\ref{def:Gabidulin-codes}, where $\mathcal{A}=\{\alpha_1,\ldots,\alpha_n\}\subseteq \mathbb F_{q^m}$ is linearly independent over $\F_q$, and let $V:=\mathrm{span}_{\mathbb F_q}\{\alpha_1,\ldots,\alpha_n\}$. Recall that $\cC$ has minimum rank distance $d=n-k+1.$ 

For integers $r\ge 0$ and $e>0$ , define 

\[\begin{split}
\sigma_r:=\sum_{i=0}^{r}q^i=\frac{q^{r+1}-1}{q-1},\qquad\tau_e:=(\sigma_e-1)+\sigma_e\bigl((\sigma_e-1)\sigma_{k-1}+1\bigr).
\end{split}
\]


\begin{theorem}[Correlated agreement]\label{thm:gab-half-correlated}
Let $\cC=\mathcal{G}(\mathcal{A},k)\subseteq\mathbb F_{q^m}^n$ be a Gabidulin code. For $\bu_0,\bu_1\in\mathbb F_{q^m}^n$, let $\ell_{\bu_0,\bu_1}:=\{\bu_0+\lambda \bu_1:\lambda\in\mathbb F_{q^m}\}.$ Let $e$ be a positive integer with $1\le e<d/2$. If $\Pr_{\lambda\in\mathbb F_{q^m}}\bigl[d(\bu_0+\lambda \bu_1,\cC)\le e\bigr]>\epsilon_e:=\frac{\tau_e}{q^m},$ then
\[
d\left(
\begin{pmatrix}
\bu_0\\
\bu_1
\end{pmatrix},
\cC^2\right)\le e.
\]
Consequently, every point $x\in\ell_{\bu_0,\bu_1}$ satisfies $d(x,\cC)\le e$.
\end{theorem}

We split the proof into several lemmas. For every $\by=(y_1,\ldots,y_n)\in\mathbb F_{q^m}^n$, let $\widehat{\by}\in\mathcal L(q,m,n)$ be the unique $q$-linearized polynomial with $\widehat{\by}(\alpha_i)=y_i$ for all $i\in [n]$. Since $\alpha_1,\ldots,\alpha_n$ form an $\mathbb F_q$-basis of $V$, we have $\mathrm{rk}_q(\by)=\dim_{\mathbb F_q}\widehat{\by}(V).$ For $\by_0,\by_1\in\mathbb F_{q^m}^n$,
\[
\mathrm{rk}_q
\begin{pmatrix}
\by_0\\
\by_1
\end{pmatrix}
=
\dim_{\mathbb F_q}\left\{\bigl(\widehat{\by}_0(x),\widehat{\by}_1(x)\bigr):x\in V\right\}.
\]

Define $S:=\{z\in\mathbb F_{q^m}:d(\bu_0+z\bu_1,\cC)\le e\}$ and we have $|S|>\tau_e>\sigma_e$ given by the condition of our theorem. Our  step is to globalize the Berlekamp--Welch reconstruction systems for the received words $\bu_0+z\bu_1$, $z\in S$.

\begin{lemma}\label{lem:bw-interpolation}
If $|S|>\sigma_e$, then there exist polynomials
$
A(X,Z)=\sum_{i=0}^{e}a_i(Z)X^{q^i},
B(X,Z)=\sum_{j=0}^{k+e-1}b_j(Z)X^{q^j},
$
such that for every $r\in[n]$,
\[
A(u_{0,r}+Zu_{1,r},Z)+B(\alpha_r,Z)=0.
\]
Moreover,  
$
    a_0(Z)\ne0$, $\deg_Z a_i\le \sigma_e-q^i$, and $\deg_Z b_j(Z)\le \sigma_e$. 
\end{lemma}

\begin{proof}
We fix $z\in S$. By the definition of $S$, there exists a unique $P_z\in\mathcal L(q,m,k)$ such that $\mathrm{rk}_q\bigl(\bu_0+z\bu_1-P_z(\mathcal{A})\bigr)\le e$ where $P_z(\mathcal{A}):=\bigl(P_z(\alpha_1),\ldots,P_z(\alpha_n)\bigr).$  Define
\[
W_z:=\mathrm{span}_{\mathbb F_q}\left\{u_{0,r}+zu_{1,r}-P_z(\alpha_r):r\in[n]\right\}.
\]
It is clear $\dim_{\mathbb F_q}W_z\le e.$ Let $\Lambda_z(X)=\gamma\prod_{\alpha\in W_z}(X-\alpha)$ be the subspace annihilator polynomial of $W_z$. Since the coefficient of $X$ in $\Lambda_z(X)$ is nonzero, we can choose a nonzero $\gamma\in \F_{q^m}$ to set the coefficient of $X$ in $\Lambda_z(X)$ to be $1$. Define $\Omega_z(X):=-\Lambda_z(P_z(X)).$ Then $\deg_q\Omega_z\le k+e-1$ and
\begin{equation}\label{eq:BWid}
    \Lambda_z(u_{0,r}+zu_{1,r})+\Omega_z(\alpha_r)=0,\qquad r\in[n].
\end{equation}
Assume
\[
\Lambda_z(X)=X+\sum_{i=1}^{e}\lambda_i(z)X^{q^i},\qquad\Omega_z(X)=\sum_{j=0}^{k+e-1}\omega_j(z)X^{q^j},
\]
Then, by the equations \eqref{eq:BWid} we obtain the following linear system 
\begin{equation}\label{eq:linear-system}
\underbrace{\begin{pmatrix}
(u_{0,1}+zu_{1,1})^q
& \cdots
& (u_{0,1}+zu_{1,1})^{q^e}
& \alpha_1
& \cdots
& \alpha_1^{q^{k+e-1}}
\\
(u_{0,2}+zu_{1,2})^q
& \cdots
& (u_{0,2}+zu_{1,2})^{q^e}
& \alpha_2
& \cdots
& \alpha_2^{q^{k+e-1}}
\\
\vdots
& & \vdots
& \vdots
& & \vdots
\\
(u_{0,n}+zu_{1,n})^q
& \cdots
& (u_{0,n}+zu_{1,n})^{q^e}
& \alpha_n
& \cdots
& \alpha_n^{q^{k+e-1}}
\end{pmatrix}}_{M(z)}
\begin{pmatrix}
\lambda_1(z)\\
\vdots\\
\lambda_e(z)\\
\omega_0(z)\\
\vdots\\
\omega_{k+e-1}(z)
\end{pmatrix}
=
-
\underbrace{\begin{pmatrix}
u_{0,1}+zu_{1,1}\\
u_{0,2}+zu_{1,2}\\
\vdots\\
u_{0,n}+zu_{1,n}
\end{pmatrix}}_{:=c(z)}.
\end{equation}
Thus, for every $z\in S$, this linear system  
\begin{equation}\label{eq:linear-system2}
M(z)\begin{pmatrix}
a_1(z)\\
\vdots\\
a_e(z)\\
b_0(z)\\
\vdots\\
b_{k+e-1}(z)
\end{pmatrix}=-c(z)
\end{equation}
in the unknowns $(a_1(z),\dots,b_{k+e-1}(z))$ is solvable over $\mathbb F_{q^m}$. Let $ M(Z)$ and $ c(Z)$ denote, respectively, the coefficient matrix and the right-hand-side vector obtained from~\eqref{eq:linear-system2} by replacing $z$ with the formal variable $Z$. Next we show that $M(Z)(a_1(Z),\dots,b_{k+e-1}(Z))^T=-c(Z)$ is still solvable over the rational function field $\F_{q^m}(Z)$.

Let $r_0:=\mathrm{rank}_{\mathbb F_{q^m}(Z)}M(Z).$ If $r_0=n$, then the ${\mathbb F_{q^m}(Z)}$-rank of the augmented matrix satisfies
$$\mathrm{rank}_{\mathbb F_{q^m}(Z)}\left(M(Z)\mid c(Z)\right)=\mathrm{rank}_{\mathbb F_{q^m}(Z)}M(Z)=n,$$
so $M(Z)x=-c(Z)$ is solvable over $\mathbb F_{q^m}(Z)$. Suppose now that $r_0<n$. Since every $(r_0+1)\times(r_0+1)$ minor of $M(Z)$ vanishes identically, we have $\mathrm{rank}_{\mathbb F_{q^m}}M(z)\le r_0$ for every $z\in S$. Moreover, by~\eqref{eq:linear-system}, $M(z)x=-c(z)$ is solvable and $\mathrm{rank}_{\mathbb F_{q^m}}\left(M(z)\mid c(z)\right)=\mathrm{rank}_{\mathbb F_{q^m}}M(z)\le r_0.$ It follows that every $(r_0+1)\times(r_0+1)$ minor of $\left(M(Z)\mid c(Z)\right)$ vanishes at every point of $S$.

We now bound the $Z$-degree of such a minor. In the augmented matrix $\left(M(Z)\mid c(Z)\right)$, the column $c(Z)$ has $Z$-degree at most $1$, the columns $\left((u_{0,r}+Zu_{1,r})^{q^i}\right)_{r=1}^{n}\,(1\leq i\leq e)$ have $Z$-degree at most $q^i$, and the columns $\left(\alpha_r^{q^j}\right)_{r=1}^{n}\,(0\leq j\leq k+e-1)$ do not depend on $Z$. Hence every $(r_0+1)\times(r_0+1)$ minor of $\left(M(Z)\mid c(Z)\right)$ has $Z$-degree at most $1+\sum_{i=1}^{e}q^i=\sigma_e.$ Since $|S|>\sigma_e$, all these minors vanish identically. Therefore, $\mathrm{rank}_{\mathbb F_{q^m}(Z)}\left(M(Z)\mid c(Z)\right)=\mathrm{rank}_{\mathbb F_{q^m}(Z)}M(Z)$ and hence $M(Z)x=-c(Z)$ is solvable over $\mathbb F_{q^m}(Z)$.

Choose a nonzero $r_0\times r_0$ minor of $M(Z)$, and denote its determinant by $h(Z)$. If $r_0=0$, we use the convention $h(Z)=1$. Applying Cramer's rule to the corresponding pivot columns and setting all non-pivot variables equal to zero gives a rational solution with common denominator $h(Z)$. Multiplying this solution by $h(Z)$ and letting $a_0(Z):=h(Z)$ produces polynomials $a_i(Z)$ and $b_j(Z)$ satisfying
\[
A(u_{0,r}+Zu_{1,r},Z)+B(\alpha_r,Z)=0,\qquad r\in[n].
\]
In particular, $a_0(Z)\ne 0$.

It remains to verify the degree bounds. Since the only $Z$-dependent columns of $M(Z)$ have degrees $q,q^2,\ldots,q^e,$ we have $\deg_Z a_0\le\sum_{i=1}^{e}q^i=\sigma_e-1$. For $1\le i\le e$, if the column $\left((u_{0,r}+Zu_{1,r})^{q^i}\right)_{r=1}^{n}$ is a pivot column, then the Cramer numerator defining $a_i$ is obtained by replacing this column with $c(Z)$, and hence $\deg_Z a_i\le1+\sum_{\substack{1\le h\le e\\h\ne i}}q^h=\sigma_e-q^i.$ If this column is not a pivot column, then $a_i=0$, and the same bound is trivial. Similarly, for every $0\le j\le k+e-1$, $\deg_Z b_j\le1+\sum_{i=1}^{e}q^i=\sigma_e.$ This completes the proof. 

\end{proof}

Fix the pair $(A,B)$ obtained in Lemma~\ref{lem:bw-interpolation}. For every $z\in S$, write $A_z(X):=A(X,z)$ and $B_z(X):=B(X,z).$ We first record that the global Berlekamp--Welch solution specializes to a valid composition identity for every good point.

\begin{lemma}\label{lem:specialized-BW-root}
For every $z\in S$, let $P_z\in \mathcal{L}(q,m,k)$ be the unique polynomial satisfying $ \mathrm{rk}_q\bigl(\bu_0+z\bu_1-P_z(\mathcal{A})\bigr)\le e.$ Then $A_z\circ P_z+B_z=0.$
\end{lemma}

\begin{proof}
For the received word $\by=\bu_0+z\bu_1$, the pair $(A_z,-B_z)$ satisfies the Berlekamp--Welch reconstruction equations $A_z(y_r)=(-B_z)(\alpha_r)$ for $r\in [n]$. If $(A_z,-B_z)=(0,0)$, the desired identity is trivial. Otherwise, Lemma~\ref{lem:bw-existence} applied with $t=e$ gives $-B_z=A_z\circ P_z.$ Hence $A_z\circ P_z+B_z=0$. 

\end{proof}

For every $z\in S$, write $ P_z(X)=\sum_{j=0}^{k-1} c_j(z)X^{q^j}.$ For each $0\le j\le k-1$, let $p_j(Z)\in\mathbb F_{q^m}[Z]$ be the unique polynomial of degree less than $|S|$ satisfying $p_j(z)=c_j(z)$ for $z\in S$. Define $P(X,Z):=\sum_{j=0}^{k-1}p_j(Z)X^{q^j}$ and $S(Z):=\prod_{z\in S}(Z-z).$

By Lemma~\ref{lem:specialized-BW-root}, for every $z\in S$ we have $A(P(X,z),z)+B(X,z)=0.$ Thus, every coefficient, as a polynomial in $Z$, of $A(P(X,Z),Z)+B(X,Z)$ vanishes on all points of $S$. Hence we obtain the coefficientwise congruence
\begin{equation}\label{eq:mod-composition}
        A(P(X,Z),Z)+B(X,Z)\equiv 0\pmod{S(Z)}.
\end{equation}

We next remove the exceptional parameters at which $a_0$ vanishes. Define $S_0:=\{z\in S:a_0(z)\ne 0\}$ and $S_0(Z):=\prod_{z\in S_0}(Z-z).$ Since $\deg_Z a_0\le \sigma_e-1$, we have $|S_0|\ge |S|-(\sigma_e-1).$ Moreover, $a_0(Z)$ is invertible modulo $S_0(Z)$. The next lemma replaces the Reed--Solomon divisibility step. Instead of taking a quotient in a commutative bivariate polynomial ring, we reconstruct the coefficients of the quotient recursively in the residue ring $\mathbb F_{q^m}[Z]/(S_0(Z))$ and then choose a low-degree representative.

\begin{lemma}\label{lem:clear-BW-denominators}
There exists a $q$-linearized polynomial $\widetilde P(X,Z)= \sum_{j=0}^{k-1}\widetilde p_j(Z)X^{q^j} \in \mathbb F_{q^m}[X,Z]$ such that $\deg_Z \widetilde p_j\le(\sigma_e-1)\sigma_{k-1}+1$ for $0\leq j\leq k-1$ and $\widetilde P(X,Z)\equiv a_0(Z)^{\sigma_{k-1}}P(X,Z) \pmod{S_0(Z)}.$ Consequently, for every $z\in S_0$, $\widetilde P(X,z)=a_0(z)^{\sigma_{k-1}}P_z(X).$
\end{lemma}

\begin{proof}
Reducing \eqref{eq:mod-composition} modulo $S_0(Z)$ and comparing the coefficient of $X^{q^j}$ for $0\le j\le k-1$, gives
\begin{equation}\label{eq:triangular-recurrence-mod}
    a_0(Z)p_j(Z)+\sum_{r=1}^{\min\{e,j\}}a_r(Z)p_{j-r}(Z)^{q^r} + b_j(Z)\equiv 0\pmod{S_0(Z)}.
\end{equation}
Since $a_0(Z)$ is a unit modulo $S_0(Z)$, this triangular recurrence determines the residue class of each $p_j(Z)$ recursively. We now clear denominators in a way that preserves explicit degree bounds.

Define $N_j(Z)\in\mathbb F_{q^m}[Z]$ by $N_j:=-b_j a_0^{\sigma_j-1} -\sum_{r=1}^{\min\{e,j\}} a_r a_0^{\sigma_{r-1}-1}N_{j-r}^{q^r}$ for $0\le j\le k-1$, where the empty sum is interpreted as zero. We claim that, for every $0\le j\le k-1$, 
$$ N_j(Z) \equiv a_0(Z)^{\sigma_j}p_j(Z)\pmod{S_0(Z)}.$$ The proof is by induction on $j$. For $j=0$, the recurrence \eqref{eq:triangular-recurrence-mod} gives $a_0p_0+b_0\equiv 0\pmod{S_0(Z)},$ so $N_0=-b_0\equiv a_0p_0=a_0^{\sigma_0}p_0\pmod{S_0(Z)}.$ Assume the claim holds for all smaller indices. Using $\sigma_{r-1}+q^r\sigma_{j-r}=\sigma_j,$ we obtain, modulo $S_0(Z)$,
\[
\begin{aligned}
        N_j &\equiv-b_ja_0^{\sigma_j-1}- \sum_{r=1}^{\min\{e,j\}} a_r a_0^{\sigma_{r-1}-1} \bigl(a_0^{\sigma_{j-r}}p_{j-r}\bigr)^{q^r}=-a_0^{\sigma_j-1}\left(b_j+\sum_{r=1}^{\min\{e,j\}}a_r p_{j-r}^{q^r} \right) \equiv a_0^{\sigma_j}p_j,
\end{aligned}
\]
where the last congruence follows from \eqref{eq:triangular-recurrence-mod}. This proves the claim.

Next, we prove the degree bound $\deg_Z N_j\le (\sigma_e-1)\sigma_j+1.$ For the first term in the definition of $N_j$, Lemma~\ref{lem:bw-interpolation} gives $\deg_Z\bigl(b_ja_0^{\sigma_j-1}\bigr)\le (\sigma_e-1)\sigma_j+1.$ Assume the bound is known for all smaller indices. For any $1\le r\le \min\{e,j\}$, the corresponding summand has degree at most $(\sigma_e-q^r)+(\sigma_{r-1}-1)(\sigma_e-1)+q^r\bigl((\sigma_e-1)\sigma_{j-r}+1\bigr)=(\sigma_e-1)\sigma_j+1.$ Thus, $\deg_Z N_j\le (\sigma_e-1)\sigma_j+1$ for all $0\le j\le k-1$.

Finally, define $\widetilde p_j(Z):= a_0(Z)^{\sigma_{k-1}-\sigma_j}N_j(Z)$ for $ 0\le j\le k-1,$ and $\widetilde P(X,Z) :=\sum_{j=0}^{k-1}\widetilde p_j(Z)X^{q^j}.$ Then $\deg_Z\widetilde p_j\le (\sigma_{k-1}-\sigma_j)(\sigma_e-1)+\bigl((\sigma_e-1)\sigma_j+1\bigr) = (\sigma_e-1)\sigma_{k-1}+1.$ Moreover, by $N_j(Z)\equiv a_0^{\sigma_j}(Z)p_j(Z) \pmod{S_0(Z)}$, we have $\widetilde p_j(Z) \equiv a_0(Z)^{\sigma_{k-1}}p_j(Z)  \pmod{S_0(Z)}.$ Hence, $\widetilde P(X,Z)\equiv a_0(Z)^{\sigma_{k-1}}P(X,Z) \pmod{S_0(Z)}.$ Evaluating at any $z\in S_0$ gives $\widetilde P(X,z)= a_0(z)^{\sigma_{k-1}}P(X,z)  = a_0(z)^{\sigma_{k-1}}P_z(X).$

\end{proof}

The preceding lemma packages the pointwise decoded polynomials $\{P_z\}_{z\in S_0}$ into a single bounded-degree polynomial family $\widetilde P(X,Z)$. The following coefficient reconstruction is the Gabidulin substitute for the Reed--Solomon divisibility step. For each fixed $z\in S_0$, the condition $\mathrm{rk}_q(P_z(\mathcal{A})-\bu_0-z\bu_1)\le e$ means that the $\mathbb F_q$-linear map $x\mapsto P_z(x)-\widehat{\bu}_0(x)-z\widehat{\bu}_1(x)$ has kernel of dimension at least $n-e$. This subspace may depend on $z$. The next lemma shows that, if $S_0$ is large enough, there is a common $(n-e)$-dimensional subspace on which the agreement holds identically as a polynomial in $Z$.

\begin{lemma}\label{lem:common-kernel}
Let $\widetilde P(X,Z)$ be the $q$-linearized polynomial obtained in Lemma~\ref{lem:clear-BW-denominators}. If $|S_0|>\sigma_e((\sigma_e-1)\sigma_{k-1}+1),$ then there exists an $\mathbb F_q$-subspace $E\subseteq V$ such that $\dim_{\mathbb F_q}E=n-e$ and for every $x\in E$,
\begin{equation}\label{eq:P(x,Z)}
\widetilde P(x,Z)=a_0(Z)^{\sigma_{k-1}}\bigl(\widehat{\bu}_0(x)+Z\widehat{\bu}_1(x)\bigr).
\end{equation}
\end{lemma}

\begin{proof}
For every $x\in V$, define
\[
H_x(Z):=\widetilde P(x,Z)-a_0(Z)^{\sigma_{k-1}}\bigl(\widehat{\bu}_0(x)+Z\widehat{\bu}_1(x)\bigr).
\]
By Lemma~\ref{lem:clear-BW-denominators} and the bound $\deg_Za_0\le\sigma_e-1,$ we have $\deg_ZH_x\le (\sigma_e-1)\sigma_{k-1}+1.$

For every $z\in S_0$, Lemma~\ref{lem:clear-BW-denominators} gives $H_x(z)=a_0(z)^{\sigma_{k-1}} \bigl(P_z(x)-\widehat{\bu}_0(x)-z\widehat{\bu}_1(x)\bigr).$ Since the evaluation vector $P_z(\mathcal{A})-\bu_0-z\bu_1$ has rank at most $e$, the map $\phi_z:x\mapsto P_z(x)-\widehat{\bu}_0(x)-z\widehat{\bu}_1(x)$ has kernel dimension at least $n-e$.

We now show that every $(e+1)$-dimensional $\mathbb F_q$-subspace $W\subseteq V$ contains a nonzero element $x$ such that $H_x(Z)\equiv0.$ Fix such a subspace $W$. For each $z\in S_0$, $\dim_{\mathbb F_q}W+\dim_{\mathbb F_q}\ker\varphi_z\ge(e+1)+(n-e)>n.$ Thus, $W\cap\ker\varphi_z\ne\{0\}.$ Hence, for every $z\in S_0$, there exists a nonzero $x\in W$ such that $H_x(z)=0.$

Let $R_W$ be a set of representatives of the one-dimensional $\mathbb F_q$-subspaces of $W$. Then $|R_W|=\frac{q^{e+1}-1}{q-1}=\sigma_e.$ Since the map $x\longmapsto H_x(Z)$ is $\mathbb F_q$-linear, replacing a nonzero element of $W$ by another nonzero scalar multiple does not affect whether its value at $z$ is zero. Hence, for every $z\in S_0$, at least one factor in $\prod_{x\in R_W}H_x(Z)$ vanishes at $Z=z$. Therefore, $\prod_{x\in R_W}H_x(z)=0$ for $z\in S_0$. The degree of this product is at most $\sigma_e((\sigma_e-1)\sigma_{k-1}+1)<|S_0|.$ It follows that $\prod_{x\in R_W}H_x(Z)$ is the zero polynomial. Since $\mathbb F_{q^m}[Z]$ is an integral domain, there exists some nonzero $x\in W$ such that $H_x(Z)\equiv0.$

Define $E_*:=\{x\in V:H_x(Z)\equiv0\}.$ Since $x\mapsto H_x(Z)$ is $\mathbb F_q$-linear, $E_*$ is an $\mathbb F_q$-subspace of $V$. We have proved that every $(e+1)$-dimensional subspace $W\subseteq V$ intersects $E_*$ nontrivially. Suppose, towards a contradiction, that $\dim_{\mathbb F_q}E_*\le n-e-1.$ Then $\dim_{\mathbb F_q}E_*+(e+1)\le n,$ so there exists an $(e+1)$-dimensional subspace $W\subseteq V$ such that $W\cap E_*=\{0\},$ contradicting the preceding conclusion. Therefore, $\dim_{\mathbb F_q}E_*\ge n-e.$

Choose any $(n-e)$-dimensional subspace $E\subseteq E_*.$ Then $H_x(Z)\equiv0$ for every $x\in E$, which is exactly
\[
\widetilde P(x,Z)=a_0(Z)^{\sigma_{k-1}}\bigl(\widehat{\bu}_0(x)+Z\widehat{\bu}_1(x)\bigr),\qquad x\in E \ . 
\]
\end{proof}

\begin{proof}[The proof of Theorem~\ref{thm:gab-half-correlated}]
By the assumption of the theorem,
\[
|S|=q^m\Pr_{\lambda\in\mathbb F_{q^m}}\bigl[d(\bu_0+\lambda \bu_1,\cC)\le e\bigr]>\tau_e.
\]
Recall that $\tau_e=(\sigma_e-1)+\sigma_e((\sigma_e-1)\sigma_{k-1}+1).$ Since $|S_0|\ge |S|-(\sigma_e-1),$ we obtain $|S_0|>\sigma_e((\sigma_e-1)\sigma_{k-1}+1).$ Therefore, Lemma~\ref{lem:common-kernel} gives an $\mathbb F_q$-subspace $E\subseteq V$ such that $\dim_{\mathbb F_q}E=n-e$ and $\widetilde P(x,Z)=a_0(Z)^{\sigma_{k-1}}\bigl(\widehat{\bu}_0(x)+Z\widehat{\bu}_1(x)\bigr)$ for all $x\in E$.

Since $k+2e\le n$, we have $n-e\ge k+e\ge k.$ Thus, we can choose $\mathbb F_q$-linearly independent elements $x_1,\ldots,x_k\in E.$ By Moore interpolation, there exist unique polynomials $P_0,P_1\in\mathcal L(q,m,k)$ such that
\[
P_s(x_i)=\widehat{\bu}_s(x_i), \qquad i\in[k],\ s\in\{0,1\}.
\]

We claim that
\begin{equation}\label{eq:P(X,Z)}
\widetilde P(X,Z)=a_0(Z)^{\sigma_{k-1}}\bigl(P_0(X)+ZP_1(X)\bigr).
\end{equation}
Indeed, both sides are $q$-linearized polynomials in $X$ of $q$-degree less than $k$ over the field $\mathbb F_{q^m}(Z)$. Moreover, by~\eqref{eq:P(x,Z)} and the choice of $P_0,P_1$, they agree at the $k$ linearly independent points $x_1,\ldots,x_k.$ Hence, they must be identical.

Now let $x\in E$. Combining~\eqref{eq:P(x,Z)} and~\eqref{eq:P(X,Z)}, we obtain
\[
a_0(Z)^{\sigma_{k-1}}\bigl(P_0(x)+ZP_1(x)\bigr)=a_0(Z)^{\sigma_{k-1}}\bigl(\widehat{\bu}_0(x)+Z\widehat{\bu}_1(x)\bigr).
\]
Since $a_0(Z)\ne0$, we may cancel $a_0(Z)^{\sigma_{k-1}}$. Comparing the coefficients of $1$ and $Z$, we obtain $P_0(x)=\widehat{\bu}_0(x)$ and $P_1(x)=\widehat{\bu}_1(x)$ for all $x\in E$.

Let $\mathbf{c}_s:=P_s(\mathcal{A})=(P_s(\alpha_1),\ldots,P_s(\alpha_n))\in \cC$ with $s\in \{0,1\}$. Thus, the $\mathbb F_q$-linear map
\[
V\longrightarrow \mathbb F_{q^m}^2,
\qquad
x\longmapsto
\bigl(\widehat{\bu}_0(x)-P_0(x),\widehat{\bu}_1(x)-P_1(x)\bigr)
\]
vanishes on $E$. Hence its rank is at most $n-\dim_{\mathbb F_q}E=e.$ Since $\alpha_1,\ldots,\alpha_n$ form an $\mathbb F_q$-basis of $V$, this gives
\[
d\left(
\begin{pmatrix}
\bu_0\\
\bu_1
\end{pmatrix},
\cC^2
\right)
\leq e.
\]
This proves the correlated-agreement assertion. It remains to verify the final consequence.

Finally, for every $\lambda\in\mathbb F_{q^m}$, we have $\bc_0+\lambda \bc_1\in \cC$ and
\[
    \mathrm{rk}_q\bigl(\bu_0+\lambda \bu_1-(\bc_0+\lambda \bc_1)\bigr)\le\mathrm{rk}_q
        \begin{pmatrix}
        \bu_0-\bc_0\\
        \bu_1-\bc_1
        \end{pmatrix}
        \le e.
\]
Thus $d(\bu_0+\lambda \bu_1,\cC)\le e$ for all $\lambda$.

\end{proof}

\begin{remark} 
Recall that $\tau_e=(\sigma_e-1)+\sigma_e\bigl((\sigma_e-1)\sigma_{k-1}+1\bigr)$. Since $\sigma_i\le \frac{q}{q-1} q^i$, we obtain $\tau_e\le (\frac{q}{q-1})^3q^{k+2e-1}+2(\frac{q}{q-1}) q^e.$ As $k+2e-1\le n-1$, 
\[
\begin{aligned}
\tau_e\le(\frac{q}{q-1})^3\cdot q^{n-1}+2(\frac{q}{q-1})\cdot q^{n-2}=\left[\left(\frac{q}{q-1}\right)^3+\frac{2}{q-1}\right]q^{n-1}\le10q^{n-1}.
\end{aligned}
\]
Consequently, the hypothesis of Theorem~\ref{thm:gab-half-correlated} is implied by the more explicit condition $|S|>10q^{n-1}$.
\end{remark}

\subsection{Counterexamples at the Unique-Decoding Boundary}\label{sec:limitations-radius}
Theorem~\ref{thm:gab-half-correlated} establishes a correlated agreement for Gabidulin codes for every integer radius $e<d/2$. We now show that the strict inequality is necessary: at the boundary $e=d/2$, no correlated-agreement statement with soundness error less than one can hold in general. More precisely, there exists an infinite family of Gabidulin codes of constant rate $\frac{1}{3}$ such that there exists an affine line $\bu+z\bv$ where all points are $\frac{d}{2}$-close to the Gabidulin code but $\bv$ is at least $\frac{3d}{4}$ away from this code.

We first recall the character-sum estimate used in the construction. This is a standard form of Weil's bound for multiplicative character sums.

\begin{lemma}[{\cite[Theorem~5.41]{LN97}}]\label{lem:weil-bound}
Let $\mathbb F$ be a finite field with $|\F|=M$, and let $\psi: \F^*\to \CC^*$ be a nontrivial multiplicative character of order $\ell>1$. Extend $\psi$ to $\mathbb F$ by setting $\psi(0)=0$. Let $f\in \mathbb F[X]$ be nonconstant, and suppose that $f$ is not of the form $c g(X)^\ell$ for any $c\in \mathbb F^*$ and $g(X)\in \mathbb F[X]$. Let $s$ denote the number of distinct roots of $f$ in its splitting field. Then
\[
\left|\sum_{x\in \mathbb F}\psi(f(x))\right| \le (s-1)\sqrt{M}.
\]
\end{lemma}

Throughout this subsection, fix an integer $g\ge 1$, and put $Q=q^g,\,F=\mathbb F_{Q^6}=\mathbb F_{q^{6g}}.$ Thus, $F$ is a $6$-dimensional vector space over $\mathbb F_Q$. For an $\F_Q$-subspace $U\subseteq F$, we write $A_U(X):=\prod_{u\in U}(X-u)$ for its subspace-annihilator polynomial. If $\dim_{\F_Q}U=2$, then $A_U(X)=X^{Q^2}+\Lambda(U)X^Q+\mu X.$ If $\dim_{\F_Q}W=4$, then $A_W(X)=X^{Q^4}+\Lambda(W)X^{Q^3}+\text{lower $Q$-degree terms}.$ For any $\mathbb F_Q$-subspace $U\subseteq F$, we write $\Lambda(U)$ for the coefficient of the second highest $Q$-degree term in its subspace-annihilator polynomial $A_U$.

\begin{lemma}\label{lem:two-dim-diversity}
With the notation introduced above, define $\mathcal L_2:=\{\Lambda(U): U\subseteq F,\ \dim_{\F_Q}U=2\}$. Then $\mathcal L_2=F.$
\end{lemma}

\begin{proof}
First, taking $U=\mathbb F_{Q^2}$ gives $A_U(X)=X^{Q^2}-X$, so $0\in\mathcal L_2$. Every two-dimensional $\mathbb F_Q$-subspace of $F$ can be written in the form $\mathrm{span}_{\mathbb F_Q}\{x,xr\}$ for some $x\in F^*$ and $r\in F\setminus \mathbb F_Q$. To obtain all nonzero elements of $F$, it is enough to consider the subfamily with $r\in F\setminus \mathbb F_{Q^2}.$

For $x\in F^*$ and $r\in F\setminus \mathbb F_{Q^2}$, define $U_{x,r}:=\mathrm{span}_{\mathbb F_Q}\{x,x r\}.$ Then $U_{x,r}$ is a two-dimensional $\mathbb F_Q$-subspace of $F$. Write $A_{U_{x,r}}(X)=X^{Q^2}+\Lambda(U_{x,r})X^Q+\mu X.$ Since both $\mathbf{x}$ and $\mathbf{x}\mathbf{r}$ are roots of $A_{U_{\mathbf{x},\mathbf{r}}}$, we have $x^{Q^2}+\Lambda(U_{x,r})x^Q+\mu x=0$ and $r^{Q^2}x^{Q^2}+\Lambda(U_{x,r})r^Qx^Q +\mu rx=0.$ Subtracting $r$ times the first identity from the second gives $(r^{Q^2}-r)x^{Q^2}+\Lambda(U_{x,r})(r^Q-r)x^Q=0.$ Since $r\notin\mathbb F_Q$, we have $r^Q-r\ne 0$. Hence
\[
    \Lambda(U_{x,r})=-x^{Q^2-Q}\cdot\frac{r^{Q^2}-r}{r^Q-r}.
\]
Let $R(r):=\frac{r^{Q^2}-r}{r^Q-r}.$ Since $r\in F\setminus \mathbb F_{Q^2}$, we have $R(r)\ne 0$.

Fix $r\in F\setminus \mathbb F_{Q^2}$. As $x$ varies over $F^*$, the element $x^{Q^2-Q}$ ranges over $\ker N_{F/\mathbb F_Q} := \{a\in F^*: N_{F/\mathbb F_Q}(a)=1\}$ where $N_{F/\mathbb F_Q}:F^*\to \mathbb F_Q^*,\,\, a\mapsto a^{1+Q+\cdots+Q^5},$ denotes the field norm from $F$ to $\mathbb F_Q$. Indeed, $N_{F/\mathbb F_Q}(x^{Q^2-Q})=1$, and the image has size 
\[
\frac{|F^*|}{\gcd(Q^6-1,Q^2-Q)} = \frac{Q^6-1}{Q-1} = |\ker N_{F/\mathbb F_Q}|.
\]
Therefore, for every fixed $r\in F\setminus \mathbb F_{Q^2}$, we have
\[
    \{\Lambda(U_{x,r}):x\in F^*\}=-R(r)\cdot \ker N_{F/\mathbb F_Q} = \{a\in F^*:N_{F/\mathbb F_Q}(a)=N_{F/\mathbb F_Q}(-R(r))\}.
\]
Thus, as $r$ varies, two parameters $r$ and $r'$ contribute the same coset precisely when $N_{F/\mathbb F_Q}(-R(r))=N_{F/\mathbb F_Q}(-R(r')).$ Cosets corresponding to distinct norm values are disjoint. Consequently, to show that $F^*\subseteq\mathcal L_2$, it remains to prove that $N_{F/\mathbb F_Q}(-R(r))$ ranges over all of $\mathbb F_Q^*$ as $r$ ranges over $F\setminus\mathbb F_{Q^2}$.

We will show that every nonzero norm value occurs. Let $\chi:\mathbb F_Q^*\to \mathbb C^*$ be a nontrivial multiplicative character, and let
\[ \widetilde \chi:=\chi\circ N_{F/\mathbb F_Q}:F^*\to \mathbb C^*.
\]
Since the norm map is surjective, $\widetilde \chi$ is a nontrivial multiplicative character of $F^*$. Extend it to $F$ by setting $\widetilde\chi(0)=0$. Let $R(X):=(X^{Q^2}-X)/(X^Q-X)$ denote the polynomial obtained from the expression $R(r)$ by replacing $r$ with the variable $X$ and canceling the common factor $X^Q-X$. Then $\deg R(X)=Q^2-Q$, and the roots of $R(X)$ are exactly $\mathbb F_{Q^2}\setminus \mathbb F_Q$, all of which are simple. In particular, $R(X)$ is not a nontrivial power. Thus, by Lemma~\ref{lem:weil-bound}, we have
\[
\left| \sum_{r\in F}\widetilde \chi(R(r)) \right| \le (Q^2-Q-1)\sqrt{|F|}=(Q^2-Q-1)Q^3.
\]
Since $R(X)=1+(X^Q-X)^{Q-1}$, we have $R(r)=1$ for $r\in\mathbb F_Q$. On $\mathbb F_{Q^2}\setminus\mathbb F_Q$, we have $R(r)=0$. Removing the points in $\mathbb F_{Q^2}$ therefore changes the sum by $Q$. Hence, for every nontrivial multiplicative character $\chi$ of $\mathbb F_Q^*$,
\[
\left|\sum_{r\in F\setminus\mathbb F_{Q^2}}
\chi\!\left(N_{F/\mathbb F_Q}(-R(r))\right)\right|
\le (Q^2-Q-1)Q^3+Q.
\]

We now use this character-sum bound to count the occurrences of each norm value. Fix $t\in\mathbb F_Q^*$, and write $M_t:=\left|\{r\in F\setminus\mathbb F_{Q^2}:N_{F/\mathbb F_Q}(-R(r))=t\}\right|.$ By the orthogonality of multiplicative characters of $\mathbb F_Q^*$, we have
\[
M_t=\frac{1}{Q-1}\sum_{\chi}\chi(t^{-1})
\sum_{r\in F\setminus\mathbb F_{Q^2}}
\chi\!\left(N_{F/\mathbb F_Q}(-R(r))\right),
\]
where the outer sum is over all multiplicative characters of $\mathbb F_Q^*$. The trivial character contributes $Q^6-Q^2$, while each of the $Q-2$ nontrivial characters has absolute contribution at most $(Q^2-Q-1)Q^3+Q$. Therefore,
\[
\begin{aligned}
M_t&\ge\frac{Q^6-Q^2-(Q-2)\bigl((Q^2-Q-1)Q^3+Q\bigr)}{Q-1}\\
&=Q(3Q^3+2Q^2-2)>0.
\end{aligned}
\]
Thus every $t\in\mathbb F_Q^*$ occurs. Since varying $x\in F^*$ gives all elements in the corresponding norm coset, every nonzero element of $F$ belongs to $\mathcal L_2$. Together with $0\in\mathcal L_2$, we obtain $\mathcal L_2=F$.

\end{proof}

\begin{lemma}\label{lem:lift-to-four}
With the notation introduced above, define $\mathcal L_4:=\{\Lambda(V): V\subseteq F,\ \dim_{\mathbb F_Q}V=4\}.$ Then $\mathcal L_4=F.$
\end{lemma}

\begin{proof}
Let $U\subseteq F$ be a $2$-dimensional $\F_Q$-subspace, and write $A_U(X)=X^{Q^2}+\alpha X^Q+\mu X.$ Set $ W:=A_U(F)=\{A_U(x):x\in F\}.$ The map $A_U:\ F\to F$ is $\F_Q$-linear and has kernel $U$. Since $\dim_{\F_Q}F=6$ and $\dim_{\F_Q}U=2$, the image $W$ has dimension $4$ over $\F_Q$.

Write $A_W(X)=X^{Q^4}+\beta X^{Q^3}+\text{lower $Q$-degree terms}.$ Because $W=A_U(F)$, the composition $A_W(A_U(X))$ vanishes on all of $F$. It is a monic $Q$-linearized polynomial of $Q$-degree $6$, hence $A_W(A_U(X))=X^{Q^6}-X.$ Observe that the coefficient of $X^{Q^5}$ is $\alpha^{Q^4}+\beta,$ where the term $\alpha^{Q^4}$ comes from $A_U(X)^{Q^4}$ and the term $\beta$ comes from $\beta A_U(X)^{Q^3}$. Since $X^{Q^6}-X$ has no $X^{Q^5}$ term, it must hold $\beta=-\alpha^{Q^4}.$ As $U$ ranges over all $2$ dimensional subspace, the coefficient $\beta=\Lambda(W)$ ranges over $\{-\alpha^{Q^4}:\alpha\in\mathcal L_2\}.$ Since the map $\alpha\mapsto -\alpha^{Q^4}$ is a bijection of $F$ and $\mathcal L_2=F$ by Lemma~\ref{lem:two-dim-diversity}, this set is $F$. Hence $\mathcal L_4=F$.

\end{proof}

We are ready to show the counterexample at the unique-decoding boundary.

\begin{theorem}\label{thm:unique-decoding-counterexample}
For every prime power $q$, there is an infinite family of Gabidulin codes indexed by $g$ with the following property. Let $F=\F_{q^{6g}}$ and code length $n=6g$. Let $\mathcal{A}=\{\alpha_1,\ldots,\alpha_n\}$ be any $\F_q$-basis of $F$ and let $\cC=\mathcal{G}(\mathcal{A},2g+1)\subseteq F^n$ be the corresponding Gabidulin code. There exist words $\bu_0,\bu_1\in F^n$ such that
\[
  \Pr_{\lambda\in F}\bigl[d(\bu_0+\lambda \bu_1,\cC)\le d/2\bigr]=1,
\]
but $d([\bu_0,\bu_1],\cC^2)\ge d(\bu_1,\cC)\ge \frac{3d}{4}$.
\end{theorem}

\begin{proof}
Let $Q=q^g$ and let $\widehat{\bu}_0,\widehat{\bu}_1$ be the $q$-linearized polynomials
\[
  \widehat{\bu}_0(X)=X^{Q^4},\qquad \widehat{\bu}_1(X)=X^{Q^3}.
\]
Let $\bu_s:=\widehat{\bu}_s(\mathcal{A})\in F^n$ for $s\in \{0,1\}$. Since $\cC=\mathcal{G}(\mathcal{A},2g+1)$, its codewords are evaluations of $q$-linearized polynomials of $q$-degree at most $2g$.

By Lemma~\ref{lem:lift-to-four}, for every $\lambda\in F$ there is a $4$-dimensional $\F_Q$-subspace $W\le F$ with $A_W(X)=X^{Q^4}+\lambda X^{Q^3}+c_1X^{Q^2}+c_2X^Q+c_3X.$ For such a $\lambda$, define $P_\lambda(X):=-c_1X^{Q^2}-c_2X^Q-c_3X.$ This polynomial has $q$-degree at most $2g$, and hence $P_\lambda(\mathcal{A})\in \cC$. Moreover, $\widehat{\bu}_0(X)+\lambda \widehat{\bu}_1(X)-P_\lambda(X)=A_W(X).$ The kernel of $A_W$ is exactly $W$, which has $\F_q$-dimension $4g$. Since $F$ has $\F_q$-dimension $6g$, the evaluation vector $A_W(\mathcal{A})$ has rank $6g-4g=2g=d/2$.  Therefore, $d(\bu_0+\lambda \bu_1,\cC)\le 2g=d/2$ for every $\lambda\in F$, proving the claimed equality.

It remains to show that the interleaved word is farther than $d/2$ from $\cC^2$. It suffices to look at the second row. For any $P(X)\in \mathcal{L}(q,6g,2g+1)$,  $X^{q^{3g}}-P(X)$ is a nonzero $q$-linearized polynomial of $q$-degree at most $3g$. Hence its kernel in $F$ has $\F_q$-dimension at most $3g$, and its image on $F$ has $\F_q$-dimension at least $3g$. Since $\mathcal{A}$ is an $\F_q$-basis of $F$, we obtain $\rk_q\bigl(\bu_1-P(\mathcal{A})\bigr)\ge 3g.$ Consequently, for every $\bc_0,\bc_1\in \cC$,
\[
  \rk_q\begin{pmatrix}\bu_0-\bc_0\\ \bu_1-\bc_1\end{pmatrix}\ge\rk_q(\bu_1-\bc_1)\ge 3g,
\]
and therefore $d([\bu_0,\bu_1],\cC^2)\ge\frac{3d}{4}.$
\end{proof}

The same construction also gives a lower bound on the soundness error for proximity gaps. Consider the line $\bu_1+\lambda \bu_0$. Since $C$ is $F$-linear, for every $\lambda\ne0$ we have
\[
d(\bu_1+\lambda\bu_0,\cC)=d(\bu_0+\lambda^{-1}\bu_1,\cC)\le d/2,
\]
whereas $d(\bu_1,\cC)\ge3d/4>d/2$. Thus exactly $|F|-1$ parameters give close points, so any soundness error for this weaker property must be at least $1-|F|^{-1}=1-q^{-6g}$.

The preceding construction is not tied to the specific choice $F=\mathbb{F}_{Q^6}$. The same proof works whenever the ambient field is an $N$-dimensional extension of $\mathbb{F}_Q$, with $N\ge 6$. More precisely, in the proof of Lemma~\ref{lem:two-dim-diversity}, take $F=\mathbb{F}_{Q^N}$ and let $r$ range over $\{r\in F:r^{Q^2}\ne r\}$. This set has size at least $Q^N-Q^2$. The same character sum bound and orthogonality calculation give, for every $t\in\mathbb F_Q^*$,
\[
\left|\{r\in F:r^{Q^2}\ne r,\ N_{F/\mathbb F_Q}(-R(r))=t\}\right|
\ge\frac{Q^N-Q^2-(Q-2)\bigl((Q^2-Q-1)Q^{N/2}+Q\bigr)}{Q-1}>0.
\]
Indeed, the numerator divided by $Q^N$ is increasing in $N$, so positivity follows from the calculation at $N=6$ above. Hence every nonzero element of $F$ occurs as the second highest coefficient of a two-dimensional $\mathbb F_Q$-subspace annihilator polynomial. If $N$ is even, $U=\mathbb F_{Q^2}\subseteq F$ also gives coefficient zero. Applying the same composition argument as in Lemma~\ref{lem:lift-to-four} then gives all nonzero coefficients in dimension $N-2$, and all coefficients when $N$ is even. Consequently, we obtain the following wider family of counterexamples at the unique-decoding boundary.

\begin{theorem}
Fix an integer $N\ge 6$. For every prime power $q$, there is an infinite family of Gabidulin codes indexed by $g$ with the following property. Let $Q=q^g,\,F=\mathbb{F}_{Q^N}=\mathbb{F}_{q^{Ng}}$ and $n=Ng$. Let $\mathcal{A}=\{\alpha_1,\ldots,\alpha_n\}$ be any $\mathbb{F}_q$-basis of $F$, and let $\mathcal{C}=\mathcal{G}(\mathcal{A},(N-4)g+1)\subseteq F^n$ be the corresponding Gabidulin code. Moreover, there exist words $\bu_0,\bu_1\in F^n$ such that
\[
    \Pr_{\lambda\in F}\bigl[d(\bu_0+\lambda \bu_1,\cC)\le d/2\bigr]\ge 1-Q^{-N}=1-q^{-Ng},
\]
but $d([\bu_0,\bu_1],\cC^2)\ge 3g=\frac{3d}{4}.$ If $N$ is even, the probability above is equal to $1$.
\end{theorem}

\begin{proof}
Take $\bu_0=(X^{Q^{N-2}})(\mathcal A)$ and $\bu_1=(X^{Q^{N-3}})(\mathcal A)$. By the coefficient count above, the proof of Theorem~\ref{thm:unique-decoding-counterexample} gives $d(\bu_0+\lambda\bu_1,\cC)\le 2g=d/2$ for every $\lambda\ne0$, and also for $\lambda=0$ when $N$ is even. The same degree bound gives $d(\bu_1,\cC)\ge Ng-(N-3)g=3g$, proving the theorem.
\end{proof}

\subsection{A Tight Lower Bound on the Proximity Error at $\delta=\frac{d-1}{2n}$}\label{subsec:half-distance-threshold}

The counterexample in the previous subsection shows that the radius $e<d/2$ in the correlated agreement theorem cannot be relaxed to the boundary case $e=d/2$. We now give a counterexample at radius $e=(d-1)/2$, showing that the dependence on $q^{n-1}/q^m$ in the soundness error of Theorem~\ref{thm:gab-half-correlated} is tight up to a constant factor. By exchanging the two vectors, this construction also gives a line with many points close to the code but at least one point far from it.

Throughout this subsection, fix integers $r\ge 1$ and $2\le c\le h$. Define the parameters
\[
    e:=r,\quad n:=c(r+1),\quad m:=h(r+1),\quad k:=n-2r,\quad Q:=q^{r+1}.
\]
Let $E:=\mathbb F_{q^m}=\mathbb F_{Q^h}$ and let $V\subseteq E$ be any $c$ dimensional $\mathbb F_Q$-subspace. Let $\mathcal A=\{\alpha_1,\ldots,\alpha_n\}$ be an $\mathbb F_q$ basis of $V$. We write $P(\mathcal A):=(P(\alpha_1),\ldots,P(\alpha_n))$ for the evaluation of a $q$-linearized polynomial $P$ on $\mathcal A$. For an $\mathbb F_Q$ subspace $W\subseteq V$, write $A_W(X):=\prod_{w\in W}(X-w)$ for its subspace annihilator polynomial.

For each $\mathbb F_Q$ hyperplane $W\subseteq V$, its subspace annihilator polynomial $A_W$ vanishes on $W$. Thus, for every $x\in V\setminus W$, the value $A_W(x)$ is nonzero and gives a parameter $-A_W(x)^{q-1}$. The following lemma counts all parameters obtained in this way.

\begin{lemma}\label{lem:hyperplane-image-parameters}
With the notation above, define
\[
    H:= \bigl\{ -A_W(x)^{q-1} : W\subseteq V,\  \dim_{\mathbb F_Q}W=c-1,\  x\in V\setminus W \bigr\}.
\]
Then $|H|=\frac{q^n-1}{q-1}.$
\end{lemma}
\begin{proof}
For each $\mathbb F_Q$ hyperplane $W\subseteq V$, the polynomial $A_W$ is $Q$ linearized and has kernel $W$. Hence $A_W(V):=\{A_W(x):x\in V\}$ is a one dimensional $\mathbb F_Q$ subspace of $E$.

We claim that distinct hyperplanes have distinct images. Indeed, put $U:=A_W(V)$, and let $A_U$ and $A_V$ be the subspace annihilator polynomials of $U$ and $V$, respectively. The polynomial $A_U\circ A_W$ is monic of $Q$ degree $c$ and vanishes on $V$, so $A_U\circ A_W=A_V.$ If another hyperplane $W'\subseteq V$ satisfies $A_{W'}(V)=U$, then $A_U\circ A_{W'}=A_V$ as well. Subtracting the two identities gives $A_U\circ(A_W-A_{W'})=0.$ Since $A_U\neq 0$ and the $q$-degree is additive under composition, this is possible only if $A_W-A_{W'}=0$. Therefore, $W=W'$.

Fix an image $A_W(V)$. Since $A_W(V)$ is one dimensional over $\mathbb F_Q$, it has $Q-1$ nonzero elements. Consider the map $b\in A_W(V)\setminus\{0\}\longmapsto -b^{q-1}.$ Two nonzero elements $b,b'\in A_W(V)$ give the same parameter exactly when $ b^{q-1}=(b')^{q-1},$ Hence every parameter has exactly $q-1$ preimages in $A_W(V)\setminus\{0\}$ and therefore a fixed image $A_W(V)$ gives $\frac{Q-1}{q-1}$ distinct parameters.

We next show that different images give disjoint sets of parameters. Suppose $A_W(V)$ and $A_{W'}(V)$ give a common parameter. Then there exist nonzero $b\in A_W(V)$ and $b'\in A_{W'}(V)$ such that $b^{q-1}=(b')^{q-1}.$ Since $U$ and $U'$ are both one dimensional over $\mathbb F_Q$, we have $A_W(V)=\mathbb F_Q b=\mathbb F_Q b'=A_{W'}(V).$ Thus distinct images contribute disjoint sets of parameters.

Finally, a $c$-dimensional $\mathbb F_Q$-space has $\frac{Q^c-1}{Q-1}$ hyperplanes. By the preceding claim, these hyperplanes give distinct images. Therefore,
\[
    |H|=\frac{Q^c-1}{Q-1}\cdot\frac{Q-1}{q-1} =\frac{Q^c-1}{q-1} =\frac{q^n-1}{q-1},
\]
where the last equality follows from $Q^c=q^n$.
\end{proof}

\begin{theorem}\label{thm:half-distance-threshold-counterexample}
Let $\cC:=\mathcal G(\mathcal A,k)\subseteq E^n$ be the Gabidulin code over $\mathcal A$ defined above. Its minimum distance is $d=2r+1$, so $e=r=(d-1)/2$. Define $\bu_0:=(X^{q^{n-e}})(\mathcal A),$ $\bu_1:=(X^{q^{n-e-1}})(\mathcal A)$ and let $\ell(\lambda):=\bu_0+\lambda\bu_1$ for $\lambda\in E$. Then $d(\ell(\lambda),\cC)\le e$ for every $\lambda\in H$, whereas $d([\bu_0,\bu_1],\cC^2)\ge e+1.$
\end{theorem}

\begin{proof}
We first show that every $\lambda\in H$ gives an $e$ close point. By the definition of $H$, there exist an $\mathbb F_Q$ hyperplane $W\subseteq V$ and a nonzero $b\in A_W(V)$ such that $\lambda=-b^{q-1}$. Write
\[
    A_W(X)
    =X^{Q^{c-1}}+\sum_{i=0}^{c-2}b_iX^{Q^i},
    \qquad b_i\in E,
\]
and define $L(X):=A_W(X)^q+\lambda A_W(X).$ Since $(c-1)(r+1)=n-e-1$, expanding $L$ gives
\[
    L(X)
    =X^{q^{n-e}}+\lambda X^{q^{n-e-1}}
     +\sum_{i=0}^{c-2}
      \left(b_i^qX^{qQ^i}+\lambda b_iX^{Q^i}\right).
\]
Define $P_\lambda(X):=-\sum_{i=0}^{c-2} \left(b_i^qX^{qQ^i}+\lambda b_iX^{Q^i}\right).$ Since $Q^i=q^{i(r+1)}$, we have $\deg_q P_\lambda\le (c-2)(r+1)+1 =n-2r-1 =k-1.$ Thus $P_\lambda(\mathcal A)\in\cC$ and $\ell(\lambda)-P_\lambda(\mathcal A)=L(\mathcal A).$

We next compute the rank of $L(\mathcal A)$. The polynomial $A_W$ is $Q$ linearized and has kernel $W$, so $A_W(V)$ is one dimensional over $\mathbb F_Q$. Since $b\in A_W(V)\setminus\{0\}$, we have $A_W(V)=b\mathbb F_Q.$ For every $t\in\mathbb F_Q$, we have $(bt)^q+\lambda(bt)=b^q(t^q-t)$ which is zero exactly when $t\in\mathbb F_q$. Define the $\mathbb F_q$-linear map $\varphi:A_W(V)\to E$ with $\varphi(y):=y^q+\lambda y.$ Then $ \ker\varphi=b\mathbb F_q$ and $\dim_{\mathbb F_q}\ker\varphi=1$. Moreover, since $L=\varphi\circ A_W$ on $V$ and $A_W$ maps $V$ onto $A_W(V)$, we have $L(V)=\varphi(A_W(V)).$ Therefore, $\dim_{\mathbb F_q}L(V)=\dim_{\mathbb F_q}\varphi(A_W(V))=e$. As $\mathcal A$ is an $\mathbb F_q$-basis of $V$, we obtain $d(\ell(\lambda),\cC) \le \rk_q(L(\mathcal A)) =\dim_{\mathbb F_q}L(V) =e.$ By Lemma~\ref{lem:hyperplane-image-parameters}, there are $(q^n-1)/(q-1)$ parameters in $H$. Hence
\[
    \Pr_{\lambda\in E}
    \bigl[d(\ell(\lambda),\cC)\le e\bigr]
    \ge \frac{|H|}{|E|}
    =\frac{q^n-1}{(q-1)q^m}.
\]

It remains to show that $d(\bu_1,\cC)\ge e+1$, which implies $d([\bu_0,\bu_1],\cC^2)\ge e+1$. It suffices to consider $\bu_1$. For every $P\in\mathcal L(q,m,k)$, we have $\deg_q P\le k-1<n-e-1.$ Thus $X^{q^{n-e-1}}-P(X)$ is a nonzero $q$-linearized polynomial of $q$-degree $n-e-1$. Its kernel in $V$ has $\mathbb F_q$ dimension at most $n-e-1$. Since $\dim_{\mathbb F_q}V=n$ and $\mathcal A$ is an $\mathbb F_q$-basis of $V$, this gives $\mathrm{rk}_q\bigl(\bu_1-P(\mathcal A)\bigr)\ge n-(n-e-1) =e+1.$ Consequently, $d([\bu_0,\bu_1],\cC^2)\ge d(\bu_1,\cC)\ge e+1.$
\end{proof}

\begin{remark}
Theorem~\ref{thm:half-distance-threshold-counterexample} gives an affine line for which at least
\[
    \frac{|H|}{|E|}
    =\frac{q^n-1}{(q-1)q^m}
    =\Theta\!\left(\frac{q^{n-1}}{q^m}\right)
\]
of its points are $e$ close to $\cC$, while $d([\bu_0,\bu_1],\cC^2)\ge e+1.$ Hence, the $q^{n-1}/q^m$ dependence in Theorem~\ref{thm:gab-half-correlated} is tight up to a constant factor.
\end{remark}

%% file: sections/ligero_rm.tex
\section{Cryptographic applications of proximity gap}

\subsection{Proximity test for interleaved rank metric codes}


Recall from Definition~\ref{def:Interleaved-code} that, for an $[n,k,d]$ $\mathbb F_{q^m}$-linear rank-metric code $\mathcal{C}$, the interleaved code $\cC^s$ consists of all $s\times n$ matrices whose rows are codewords in $\cC$. In this subsection, we mainly consider the case $\cC=\mathcal{G}(\mathcal{A},k)$ and view a purported interleaved codeword as $U^T$ where $U\in \mathbb F_{q^m}^{n\times s}$; equivalently, $U^T\in \cC^s$ if and only if every column of $U$ is a Gabidulin codeword.

We present an interactive oracle proof of proximity for such interleaved Gabidulin codewords. The test reduces the interleaved membership condition to membership in the underlying Gabidulin code by checking a random $\F_{q^m}$-linear combination of the columns of $U$. Specifically, given a purported interleaved Gabidulin codeword $U^T$ with oracle query access, an honest prover $\mathcal{P}$ can, after several rounds of interaction, convince a verifier $\mathcal{V}$ that $U^T\in \cC^s$. At the same time, any malicious prover can cause $\mathcal{V}$ to accept an invalid input $(U^*)^T$ that is far from $\cC^s$ only with negligible probability.


\begin{mybox}[]{IOPP for Interleaved Rank-Metric Codes}\label{proc:IOPP}

\textbf{Oracle:} Let $U\in \mathbb F_{q^m}^{n\times s}$ be a matrix whose $j$-th column $U[j]$ is claimed to be a Gabidulin codeword for $j\in[s]$. The verifier is given oracle access to all rows of $U$: for any $i\in [n]$, the oracle returns the $i$-th row $U(i)\in\F_{q^m}^s$ of $U$.
\\
\textbf{Proximity testing:}
\begin{enumerate}
    \item $\V\rightarrow \Po:$ a random vector $\mathbf{a}=(a_1,a_2,\dots,a_s)\in\F_{q^m}^s$.
    \item $\Po\rightarrow \V:$ $\bu=\mathbf{a}\cdot U ^T\in\F_{q^m}^{n}$.
    \item $\V$ samples $t$ positions $\{j_1,\dots,j_t\}\subset [n]$ uniformly at random, and queries $U(j_{\ell})\in \mathbb F_{q^m}^s$ for all $\ell\in [t]$.
    \item $\V$ accepts if and only if $\bu$ is a Gabidulin codeword and $\bu_{j_{\ell}}=\ba\cdot U(j_{\ell})^T$ for all $\ell\in [t]$.
\end{enumerate}
    
\end{mybox}

\begin{lemma}\label{lem:IGd/2}
 {Let $\cC=\mathcal{G}(\cA,k)\subseteq\mathbb F_{q^m}^n$ be an $[n,k,d]$ Gabidulin code with dimension $k$ and rank distance $d=n-k+1$}. Let $e$ be a positive integer with $e\leq \lfloor (d-1)/2\rfloor $ and assume $q^m > q^{e+1}$. Let $U$ be a matrix in $\mathbb F_{q^m}^{n\times s}$ and denote by $L$ the column span of $U$ over $\F_{q^m}$. 
If $d(U^T,\cC^{s}) > e,$ then there exists $\bu \in L$ such that $d(\bu,\cC) > e$.
\end{lemma}
\begin{proof}
We prove by contradiction. Assume that every vector in $L$ is within rank distance at most $e$ from $\mathcal C$. That is, for every $\ba=(a_1,\ldots,a_s)\in {\F^s_{q^m}}$, the vector $\bu_{\ba}:=\ba U^T\in {\F^n_{q^m}}$
satisfies $d(\bu_{\ba},\mathcal C)\le e.$ Since $e<d/2$, the closest codeword of $\mathcal C$ to $\bu_{\ba}$ is unique, denoted by $\bc_{\ba}$. We claim that the map $
\ba\mapsto \bc_{\ba}$ is $\F_{q^m}$-linear, namely, for all $\ba,\bb\in \F_{q^m}^s$ and $\lambda\in \F_{q^m}$, there is 
$$\bc_{\ba+\lambda \bb}=\bc_{\ba}+\lambda\bc_{\bb}.$$

To prove this claim, note that for any $\bu_{\ba},\bu_{\bb}\in L$, by assumption $\mathrm{Pr}_{\lambda\in\F_{q^m}}[d(\bu_{\ba}+\lambda \bu_{\bb},C)\le e]=1$. Then, according to $\frac{d-1}{2n}$-proximity gap of Gabidulin codes, the unique codewords $\bc_{\ba}$ and $\bc_{\bb}$ satisfy 
\[\rk_q\begin{pmatrix}
    \bu_{\ba}-\bc_{\ba}\\ \bu_{\bb}-\bc_{\bb}
\end{pmatrix}\le e.\]
Hence $d(\bu_{\ba}+\lambda\bu_{\bb},\bc_{\ba}+\lambda\bc_{\bb})\le e$ and thus the $\bc_{\ba+\lambda \bb}=\bc_{\ba}+\lambda\bc_{\bb}$.

 {
Let $\be_1,\be_2,\dots,\be_s$ be the standard basis of $\F_{q^m}^s$. By above arguments, let $\bc_{\be_i}\in\cC$ be the corresponding codeword such that $d(\bu_{\be_i},\bc_{\be_i})\leq e$ and 
\[
C:=
\begin{pmatrix}
\bc_{\be_1} \\
\vdots\\
\bc_{\be_s}
\end{pmatrix}
\in\mathbb F_{q^m}^{s\times n}.
\]
On one hand, for every $\ba=\sum_{i=1}^sa_i\be_i\in\F_{q^m}^s$, by the above linearity, $\bc_{\ba}=\ba C$. Thus,
\begin{equation}\label{eq:contradiction1}
\rk(\ba(U^T-C))=d({\ba}U^T,{\ba}C)=d(\bu_{\ba},\bc_{\ba})\leq e.
\end{equation}
 On the other hand, since $C\in\cC^s$, we have $\rk(U^T-C)\ge d(U^T,\cC^s)\ge e+1$. Let $W$ be a subspace of the column space of $U^T-C$ of dimension $e+1$. Then for any $\ba\in\F_{q^m}^s$ such that $\rk(\ba(U^T-C))\le e$, the inner product map $\phi_\ba:\ \bw\in W \mapsto \ba\bw^T\in\F_{q^m}$ has a nontrivial kernel. For each fixed $0\ne\bw\in W$, the equation
$\ba\bw^T=0$ has exactly $q^{m(s-1)}$ solutions
$\ba\in\mathbb F_{q^m}^s$. Therefore, the number of
$\ba$ satisfying $\rk\bigl(\ba(U^T-C)\bigr)\le e$
is at most
\[
(q^{e+1}-1)q^{m(s-1)}<q^{ms},
\]
where the last inequality follows from $q^m>q^{e+1}$. Thus there exists some $\ba$ such that $\rk(\ba(U^T-C))> e$, which contradicts to Equation~\eqref{eq:contradiction1}.}

 
\end{proof}

\begin{theorem}\label{thm:IGtest}
       Let $\cC^{s}\subset\F_{q^m}^{s\times n}$ be an interleaved Gabidulin code with minimum distance $d$ and $e$ be a positive integer such that $e<d/2$. Let $\epsilon_e$ be defined as in Theorem~\ref{thm:gab-half-correlated}. Suppose $U^*\in\F_{q^m}^{n\times s}$ such that $d((U^*)^T,\cC^{s})>e$. Then for any malicious $\Po^*$ strategy, the oracle $U^*$ is rejected by 
     $\V$ except with $(1-e/n)^t+\epsilon_e$ probability.
\end{theorem}
\begin{proof}
Let $L^*$ denote the $\mathbb{F}_{q^m}$-column span of $U^*$. Since $d((U^*)^T,\cC^{s})>e$, Lemma~\ref{lem:IGd/2} guarantees the existence of a vector $\bu \in L^*$ such that $d(\bu,\cC) > e$. 

Moreover, $L^*$ admits a direct sum decomposition of the form $
L^* = \mathbb{F}_{q^m}\bu \oplus W
$ for some subspace $W \subseteq L^*$. Hence any random vector $\bu^* \in L^*$ can be uniquely expressed as $
\bu^* = \alpha \bu + \bw$, $\alpha \in \mathbb{F}_{q^m}, \bw \in W.$ By  {Theorem~\ref{thm:gab-half-correlated}}, for any fixed choice of $\bw$, the probability that $\alpha \bu + \bw$ lies within distance $e$ of $\cC$ is upper bounded by $\epsilon_e$. Consequently, 
$$
\Pr_{\bu^* \in L^*}\!\big[d(\bu^*, \cC) \leq e\big] \leq  {\epsilon_e}.$$

For $\bu^*=\mathbf{\lambda}(U^*)^T$, from above, we have $d(\bu^*,\cC)>e$ with probability $1- {\epsilon_e}$. In this case, suppose $\Po$ responds a forged codeword $\bc$ in Step 2 of Proximity test. We have the Hamming distance between $\bu^*$ and $\bc$ satisfying 
$$d_H(\bu^*,\bc)\geq d(\bu^*,\bc)>e.$$ 
Under this case,
    \[
    \begin{split}
    \Pr[\V\ \text{accept}\ U^*]&\leq \Pr\left[\V\ \text{accept}\ U^*\mid d(\bu^*,\cC)>e\right]\cdot\Pr[d(\bu^*,\cC)>e]\\
    &~~+\Pr[\V\ \text{accept}\ U^*\mid d(\bu^*,\cC)\leq e]\cdot\Pr[d(\bu^*,\cC)\leq e]\\
    &\leq \frac{\binom{n-e-1}{t}}{\binom{n}{t}}+ {\epsilon_e}\leq (1-e/n)^t+ {\epsilon_e} \ .  
    \end{split} 
    \]
\end{proof}

\subsection{Ligero-style Rank Metric PCS}\label{sec:lpcs}
 
 To adapt the Ligero-based PCS from the Hamming metric
 to the rank metric, it is necessary to decompose a $q$-linearized polynomial $f(x)$ in a tensor product form. 
 
 Assume $k\leq m$ and $k$ is a square $k=\kappa^2$. For any $f(x)=\sum_{i=0}^{k-1}f_ix^{q^i}\in\cL(q,m,k)$. Note that every integer $\ell$ satisfying $0\leq \ell\leq k-1$ can be written uniquely as $\ell=i\kappa+j$ for $i,j\in\kappa$. We first define $a_{\ell}=a_{i\kappa+j}=f_{i\kappa+j}^{q^{m-i\kappa}}$ and a $\kappa\times\kappa$ matrix as follows:
    \begin{equation}\label{eq:coeff}
        \mF=\begin{pmatrix}
        a_0 &a_1 &\dots a_{\kappa-1}\\
         a_\kappa &a_{\kappa+1} &\dots a_{2\kappa-1}\\
          \vdots &\vdots &\dots \vdots\\
           a_{\kappa(\kappa-1)} &a_{\kappa(\kappa-1)+1} &\dots a_{k-1}\\
    \end{pmatrix}=\begin{pmatrix}
        f_0 &f_1 &\dots f_{\kappa-1}\\
         f_\kappa^{q^{m-\kappa}} &f_{\kappa+1}^{q^{m-\kappa}} &\dots f_{2\kappa-1}^{q^{m-\kappa}}\\
          \vdots &\vdots &\dots \vdots\\
           f_{\kappa(\kappa-1)}^{q^{m-(\kappa-1)\kappa}} &f_{\kappa(\kappa-1)+1}^{q^{m-(\kappa-1)\kappa}} &\dots f_{k-1}^{q^{m-(\kappa-1)\kappa}}\\
    \end{pmatrix}.
    \end{equation}
    Define two vectors
    \[\textbf{x}_{r}:=(x,x^{q},x^{q^2},\dots,x^{q^{\kappa-1}}),\ \textbf{x}_{l}:=(x,x^{q^{\kappa}},x^{q^{2\kappa},\dots,x^{q^{(\kappa-1)\kappa}}}).
    \] 
 Moreover, we define a composite-multiplication of $\textbf{x}_l$ with vectors in $\F_{q^m}^\kappa$ as follows: 
    \[\textbf{x}_l\odot (r_0,r_1,\dots,r_{\kappa-1})^T:=\sum_{i=0}^{\kappa-1}(r_i)^{q^{i\kappa}}.\]
   Then we have the following symbolic product decomposition for $q$-linearized polynomial.
    \begin{lemma}\label{lem:comp}
        For any $q$-linearized polynomial $f(x)\in\mL(q,m,k)$,  
        \[f(x)=\textbf{x}_l\odot(\mF\cdot\textbf{x}_r^T),\]
        where $\mF$ is the coefficient matrix of $f(x)$ defined as in Equation~\eqref{eq:coeff}. Furthermore, for any $\alpha\in\F_{q^m}$, $f(\alpha)=\textbf{x}_l\odot \mF\cdot(\alpha,\alpha^q,\dots,\alpha^{q^{\kappa-1}})^T$.
    \end{lemma}
    \begin{proof}
    According to the definition of composite-multiplication $\odot$, we have
    \[
    \begin{split}
    \textbf{x}_l\odot (\mF\cdot\textbf{x}_r^T)&=\textbf{x}_l\odot\left(\sum_{j=0}^{\kappa-1}a_{j}x^{q^j},\sum_{j=0}^{\kappa-1}a_{\kappa+j}x^{q^j},\dots,\sum_{j=0}^{\kappa-1}a_{\kappa(\kappa-1)+j}x^{q^j}\right)^T\\
    &=\sum_{i=0}^{\kappa-1}\left(\sum_{j=0}^{\kappa-1} a_{i\kappa+j}x^{q^{j}}\right)^{q^{i\kappa}}=\sum_{i=0}^{\kappa-1}\sum_{j=0}^{\kappa-1}a_{i\kappa+j}^{q^{i\kappa}}x^{q^{i\kappa+j}}\\
&=\sum_{i=0}^{\kappa-1}\sum_{j=0}^{\kappa-1}f_{i\kappa+j}^{q^{m-i\kappa+i\kappa}}x^{q^{i\kappa+j}}=f(x),
    \end{split}
    \]
    where the last $``="$ follows from each coefficient $f^{q^m}_{i\kappa+j}=f_{i\kappa+j}$ in $\F_{q^m}$. Furthermore, for any $\alpha\in\F_{q^m}$, by the same analysis, we have 
    \[\textbf{x}_l\odot \left(\mF\cdot\left(\alpha,\alpha^q,\dots,\alpha^{q^{\kappa-1}}\right)^T\right)=\sum_{i=0}^{\kappa-1}\sum_{j=0}^{\kappa-1}f_{i\kappa+j}\alpha^{q^{i\kappa+j}}=f(\alpha).  
    \]
    \end{proof}
    
\paragraph{Public parameter generation.} 
The public parameters specify the finite fields, the code parameters, the evaluation domain, and the encoding algorithm, together with the hash function used for the Merkle commitments. Let $\lambda$ be the security parameter and let
\[
    k=\kappa^2\leq m,\quad
    \kappa\leq n\leq m,\quad
    n\mid m.
\]
Set $d:=n-\kappa+1$, $\rho:=\frac{\kappa}{n}\in(0,1),$ where $\rho$ is a fixed constant. Choose $0\leq e\leq\lfloor(d-1)/2\rfloor$ and select $m$ and $n$ such that
\[
    \epsilon_e
    \leq \frac{10q^{n-1}}{q^m}
    =10q^{n-1-m}
    \leq 2^{-\lambda}.
\]
For example, the last inequality holds whenever $m\geq n-1+
    \frac{\lambda+\log_2 10}{\log_2 q}.$

Since $n\mid m$, the field $\F_{q^m}$ contains a unique subfield $\F_{q^n}$ of $q^n$ elements. Let $\cA:=
    \bigl(\gamma,\gamma^q,\ldots,\gamma^{q^{n-1}}\bigr)
    \subseteq\F_{q^n}$ be a normal basis of $\F_{q^n}$. Hence, the elements of $\cA$ are $\F_q$-linearly independent. Define the $[n,\kappa,d]$ Gabidulin code $\cC:=\mathcal G(\cA,\kappa)\subseteq\F_{q^m}^n.$

Using the linearized-monomial basis
$\{x,x^q,\ldots,x^{q^{\kappa-1}}\}$, we identify
$\mathcal L(q,m,\kappa)$ with $\F_{q^m}^{\kappa}$ via
\[
    f(x)=\sum_{i=0}^{\kappa-1}f_i x^{q^i}
    \longmapsto
    \mathbf f=(f_0,\ldots,f_{\kappa-1}).
\]
For $\mathbf f\in\F_{q^m}^{\kappa}$, define
\[
    \mathsf{Enc}(\mathbf f)
    :=
    \bigl(f(\gamma),f(\gamma^q),\ldots,
          f(\gamma^{q^{n-1}})\bigr)\in\cC.
\]
After padding $\mathbf f$ with $n-\kappa$ zeros, this encoding map can be evaluated using the $q$-transform associated with the normal basis $\cA$. The algorithm of~\cite{PW18} computes this transform in $O(n\log^2 n\log\log n)$ operations over $\F_{q^m}$, which is quasi-linear in $n$.



\paragraph{Commit phase.} Assume $f(x)\in\cL(q,m,k)$ and $\mF=(a_{i\kappa+j})_{i,j\in[0,\kappa-1]}$ be the coefficient matrix defined as in Equation~\eqref{eq:coeff}. Let $U:=\enc(\mF)\in\F_{q^m}^{n\times \kappa}$ by encoding each column $\mF[j]$ into a Gabidulin codeword for $j\in[\kappa]$. In the IOP model, $U$ serves as the oracle representation of $f$, to which the verifier has query access.

To obtain a succinct cryptographic commitment of $f$, the prover arranges the $n\kappa$ entries of $U$ in a fixed order and uses them as the leaves of a Merkle tree. The prover sends only the Merkle root as its $\mathsf{com}$. Subsequent queries to entries or rows of $U$ are answered together with authenticated Merkle openings.

\paragraph{Opening phase.} Given a commitment $\mathsf{com}$, to open it to a polynomial $f\in\mathcal L(q,m,k)$, the prover sends
$f$ and a matrix $\widehat U\in\F_{q^m}^{n\times \kappa}$. The
verifier constructs the coefficient matrix $\mathcal F$ and the corresponding encoded matrix $\enc(\mathcal F)$. It accepts if and only if $\mathsf{MerkleRoot}
   \bigl(\widehat U\bigr)
   =\mathsf{com}$
and $d\bigl(\widehat U^T,\enc(\mathcal F)\bigr)< d/2$, where $d=n-\kappa+1$ is the rank distance of $\cC$. Since the matrix $\widehat U^T$ is within rank distance $<d/2$ of $\cC^{\kappa}$. Hence, whenever a valid opening exists, it determines a unique polynomial $f$.


\paragraph{Testing phase.} Upon receiving the commitment $\mathsf{com}$ that is a Merkle root of some matrix $U$, the verifier $\V $ will interact with $\Po$ to test that $U^T$ is $\frac{d-1}{2n}$-close to the interleaved Gabidulin code $\cC^{\kappa}$ by calling IOPP~\ref{proc:IOPP}.

\paragraph{Evaluation phase.} After the commitment $\mathsf{com}$ has been fixed, the evaluation point $\alpha\in\F_{q^m}$ is chosen by the verifier and sent to the prover. The prover then returns a claimed evaluation $\beta\in\F_{q^m}$; for an honest prover, $\beta=f(\alpha)$. Once $\alpha$ and $\beta$ have been specified, the tuple $(\mathsf{com},\alpha,\beta)$ constitutes the public input to $\mathsf{Eval}$, consistently with Definition~\ref{def:pc}.

Thus with public input $(\mathsf{com},\alpha,\beta)$, the evaluation phase verifies that $\beta=f(\alpha)$ for the polynomial $f$ determined by the unique
interleaved codeword close to the committed matrix. This phase is identical to the testing phase, except that the vector $\boldsymbol{\lambda}$ sent by $\V$ is replaced with $\bb=(\alpha,\alpha^{q},\ldots,\alpha^{q^{\kappa-1}})$. Let $\bv'=\bb \mF^T$. If the prover is honest, then $\bv'$ satisfies $(x,x^{q^{\kappa}},\dots,x^{q^{\kappa(\kappa-1)}})\odot\boldsymbol{v}'^T=f(\alpha)$ by Lemma~\ref{lem:comp}. 

Overall, we have the following $q$-linearized PCS:
    
\begin{mybox}[]{$\LRM$-PCS: Rank-Metric PCS for $q$-linearized Polynomials}\label{proc:lcom}
 \begin{itemize}
     \item[--] \textbf{Commit phase.} Assume $f(x)\in\cL(q,m,k)$ and $\mF=(a_{i\kappa+j})_{i,j\in[0,\kappa-1]}$ be the coefficient matrix defined as in Equation~\eqref{eq:coeff}. Let $U=(\enc(\mF[1]),\enc(\mF[2]),\dots,\enc(\mF[\kappa]))\in\F_{q^m}^{n\times \kappa}$ be the matrix obtained by encoding each column $\mF[j]$ into a Gabidulin codeword. The prover computes the $\mathsf{com}$ as the root the Merkle-tree with leaves contain all ${n\kappa}$ entries of $U$.
     \item[--] \textbf{Testing phase.} Upon receiving the commitment, $\Po $ and $\V $ run $\mathrm{IOPP}$~\ref{proc:IOPP} to prove that $ U^T\in\cC^{\kappa}$. In Step 2, instead of sending $\bu=\boldsymbol{\lambda}U^T $ directly, $\Po$ sends $\bv=\boldsymbol{\lambda}\mF^T $. The verifier then encodes $\bv$ to obtain the codeword, ensuring that the queried linear combination is a valid codeword. Specifically, 
     \begin{enumerate}
        \item $\V$ selects $\boldsymbol{\lambda}=(\lambda_0,\dots,\lambda_{\kappa-1})\leftarrow\F_{q^m}^{\kappa}$, and sends it to $\Po$.
        \item $\Po$ responds $\boldsymbol{v}:=\boldsymbol{\lambda}\cdot\mF^{\mathrm{T}}\in\F_{q^m}^{\kappa}$. 
        \item $\V$ encodes $\boldsymbol{v}$ via invoking $\enc$, and obtains a codeword $\boldsymbol{u}=(u_1,u_2,\dots,u_n)\in\cC$. Then, $\V$ selects random $J\subset [n]$, with $\lvert J\rvert=t=\Theta(\lambda)$. Here $\lambda$ is the security parameter. Then $\V$ queries the rows $U(j)$ of $U$ for all $j\in J$, corresponding to $t\kappa$ field elements. $\V$ verifies their values with the committed Merkle root using authenticated Merkle openings.
        \item $\V$ checks $u_{j}\isequal \left< U(j),\boldsymbol{\lambda}\right>$ for all $j\in J$.
    \end{enumerate}

\item[--]\textbf{Evaluation phase.} With a public input $(\mathsf{com},\alpha,\beta)$, 
    \begin{enumerate}
        \item $\Po$ computes $\boldsymbol{v}':=\boldsymbol{b}\cdot\mF^{\mathrm{T}}\in\F_{q^m}^{\kappa\times 1}$ and sends it to $\V$, where $\boldsymbol{b}:=(\alpha,\alpha^{q},\ldots,\alpha^{q^{\kappa-1}})$. 
        \item $\V$ encodes $\boldsymbol{v}'$ via invoking $\enc$, and obtains a codeword $\boldsymbol{u}'=(u_1',\dots,u_n')\in\cC$. Then, $\V$ selects random $J'\subset [n]$, with $\lvert J'\rvert=t=\Theta(\lambda)$. Then $\V$ queries the rows $U(j)$ of $U$ for all $j\in J'$. $\V$ verifies their values with the committed Merkle root using authenticated Merkle openings.
        \item $\V$ checks $u'_{j}\isequal \left<U(j),\boldsymbol{b}\right>$ for all $j\in J'$, where $\boldsymbol{b}=(\alpha,\alpha^q,\ldots,\alpha^{q^{\kappa-1}})$.
        \item If $\Po$ passes all checks above, $\V$ computes $f(\alpha):=(x,x^{q^{\kappa}},\dots,x^{q^{\kappa(\kappa-1)}})\odot\boldsymbol{v}'^T$. If $f(\alpha)=\beta$, $\V$ accepts and outputs $1$. Otherwise, $\V$ outputs $0$.
    \end{enumerate}
\end{itemize}
\end{mybox}

\begin{lemma}\label{thm:LigeroSoundness2}
Assume $k$ is a square $k=\kappa^2$ and $k\leq m$. Let $n$ be an integer such that $n\mid m$ and $\rho=\kappa/n\in (0,1]$ is a fixed constant. Let $d=n-\kappa+1$. For any positive integer $e\leq (d-1)/2$, let $\epsilon_e$ be defined as in Theorem~\ref{thm:gab-half-correlated}. The $\LRM$-PCS \ref{proc:lcom} is a polynomial commitment scheme for $q$-linearized polynomials of $q$-degree less than $k$. Its binding soundness upper bounded by $(1-\frac{e}{n})^t+\epsilon_e.$
\end{lemma}
\begin{proof}
The completeness holds by design. 

Assume $\Po$ passes all the checks in the Testing phase with probability at least $(1-e/n)^t+\epsilon_e$, by  {Theorem}~\ref{thm:IGtest}, we have $d(U^T,\cC^{\kappa})\leq e<d/2$. Let $C\in\cC^{\kappa}$ be the unique interleaved codeword such that $d(U^T,C)\leq e.$ By interpolating $\kappa$ polynomials from all rows of $C$ and applying the inverse coefficient transform in Equation~\eqref{eq:coeff}, we get a matrix $G=(g_{i\kappa+j})_{i,j=0,1,\kappa-1}$, which corresponds to a $q$-linearized polynomials $g(x)$ of $q$-degree less than $k$. Thus the commitment is bound to a unique $q$-linearized polynomial $g(x)$. 
\end{proof}

Moreover, our PCS also satisfies the evaluation binding. That is given $(\mathsf{com},\alpha,\beta)$, if $\V$ also outputs $1$ in the evaluation phase with non-negligible probability, then  $\beta=g(\alpha)$, where $g$ is the bounded polynomial of $\mathsf{com}$ in the Testing phase. More specifically,  {let $C\in\cC^{\kappa}$ be the encoded interleaved codeword of $g$ and $\bc=\bb\cdot C$}. If a prover can pass all checks in the evaluation phase with probability at least $(1-\frac{e}{n})^t$, then $\bu'=\enc(\bv')=\bc$. Otherwise, suppose $\bu'\neq\bc$. Since  {
$$\rk(\bu'-\bc)=\rk(\bu'-\bb U^T+\bb U^T-\bb C)\leq \rk(\bu'-\bb U^T)+\rk(\bb(U^T-C)) ,$$}
we have that $\rk(\bu'-\bb U^T)\geq d-\rk(\bb(U^T-C))>d/2.$ Since the verifier selects the set of rows $J'$ randomly, the probability of passing the evaluation phase will not exceed $(1-d/2n)^t<(1-\frac{e}{n})^t$. 

Thus, in Step 4 of the evaluation phase, the injectivity of the encoding map implies that
\[(x,x^{q^{\kappa}},\dots,x^{q^{\kappa(\kappa-1)}})\odot\boldsymbol{v}'=g(\alpha).  
\]

\paragraph{Knowledge soundness:} We adapt the extractability-via-efficient-decoding argument of Brakedown~\cite[Section~4.1]{brakedown}. Both constructions use a Merkle tree to commit to a matrix oracle, and Gabidulin codes admit polynomial-time unique-decoding algorithms. 

Let $\Po^*$ be a PPT prover that causes the $\LRM$-PCS verifier to accept with non-negligible probability. By the extractability of the underlying IOP-to-argument compiler~\cite{BCS16}, except with negligible probability, there exists an efficient
straight-line extractor that recovers an oracle $\widehat U\in\mathbb F_{q^m}^{n\times\kappa}$ consistent with the Merkle root sent by $\Po^*$. The compiler also induces an IOP prover strategy $\widetilde{\Po}$ whose initial oracle is $\widehat U^T$ and whose acceptance probability is non-negligible.

By our soundness analysis of IOPP~\ref{proc:IOPP}, acceptance in the testing phase implies that $d(\widehat U^T,\cC^{\kappa})\leq e<d/2$. Let $C\in \cC^{\kappa}$ be the unique codeword such that $\rk_q(U-C^T)\leq e$. Thus, every column of $U-C^T$ has rank at most $e$, i.e., every column of $\widehat U$ lies within the unique-decoding radius of $\mathcal C$. Applying a polynomial-time Gabidulin decoder~\cite{WAS13} to each column of $\widehat U$ recovers the unique interleaved Gabiduli codeword $C$. Then, applying Moore interpolation to each column of $C^T$ and applying the inverse of the coefficient transformation as in $\eqref{eq:coeff}$, we can obtain the coefficients of the $q$-linearized polynomial $g\in\mL(q,m,k)$ to which the prover is bound. According to the binding of evaluation, for an accepted $(\mathsf{com},\alpha,\beta)$ in the evaluation phase, there must be $\beta=g(\alpha)$. Thus $((\mathsf{com},\alpha,\beta),g)\in \mathcal{R}_{\mathsf{Eval}}(\mathsf{pp})$.


\paragraph{Proof size.} The protocol consists of test and evaluation phases. In each phase, the prover sends $O(t\kappa)$ field elements, together with $O(t\kappa\log k)$ hash values for Merkle authentication, where $t$ is the number of queries such $(1-e/n)^t\leq \mathrm{negl}(\lambda)$ is negligible. Since $\frac{d-1}{2n}=\frac{1}{2}(1-\rho)$ is also a constant, we can choose $t=\Theta(\lambda)$. 
Thus, the total proof consists of $O(\lambda\kappa)$ field elements and $O(\lambda\kappa\log k)$ hash values. Using $\kappa=\sqrt{k}$, and assuming that each hash value can be encoded in $O(1)$ field elements, the total proof size is $\widetilde O(\lambda\sqrt{k})$ elements of $\mathbb F_{q^m}$.

\paragraph{Prover time.} The prover time is dominated by the following two steps: 
\begin{itemize}
    \item[--] Construction of the matrix $\mF$: this step involves at most $k=\kappa^2$ computations of Frobenius maps. Under a normal basis $\{\gamma,\gamma^q,\dots,\gamma^{q^{m-1}}\}$ of $\F_{q^m}$, any element $\alpha\in\F_{q^m}$ can be represented as $\alpha=\sum_{i=0}^{m-1}a_i\gamma^{q^i}$. Then the coefficients of $\alpha^{q^j}$ is a cyclic shift of $(a_0,a_1,\dots,a_{m-1})$ since $\alpha^{q^j}=\sum_{i=0}^{m-1}a_i\gamma^{q^{i+j}}$. We count a cyclic shift of length $m$ as one operation over $\F_{q^m}$. Thus, the cost of this step is $O(k)$.
\item[--] Encoding: the computation of $C$ from $\mF$ invokes the encoding algorithm $\kappa$ times. Following \cite{PW18}, by taking $\mathcal{A}$ to be a normal basis of $\mathbb{F}_{q^n}$ as a subfield of $\F_{q^m}$, the $q$-transform and its inverse can be interpreted as evaluation and interpolation of $q$-linearized polynomial. These operations can be performed using  $O(n \log^2 n \log\log n)=\tilde{O}(n)$ operations over $\mathbb{F}_{q^m}$. Thus this step costs $\tilde{O}(n\kappa)=\tilde{O}(k)$ operations in $\mathbb{F}_{q^m}$. 
\end{itemize} 
For a constant rate $\rho=\kappa/n\in(0,1)$, we have that the total prover complexity is $O(k)+\tilde{O}(\kappa n)=\tilde{O}(k)$ operations over $\F_{q^m}$.


\paragraph{Verifier time.} The verifier time is dominated by encoding of $\bv$ and $\bv'$, and $2t$ inner products of length $\kappa$. Thus the total cost is $\widetilde{O}(n)+O(t\kappa)=\tilde{O}(\lambda \sqrt{k})$ operations.

\smallskip Overall, we have the following theorem:

\begin{theorem}\label{thm:ligerocom}
Let $\lambda$ be the security parameter, and let $\kappa,n,m$ be positive integers satisfying
\[
k=\kappa^2\leq m, \quad \kappa\leq n, \quad n\mid m.
\]
Suppose that the code rate $\rho:=\frac{\kappa}{n}\in(0,1)$ is constant. Let $d=n-\kappa+1$ be the minimum rank distance of the underlying Gabidulin code, and choose an integer $0< e\leq\left\lfloor\frac{d-1}{2}\right\rfloor.$ 

Assume that the hash function used in the Merkle commitment is collision resistant and the proximity error $\epsilon_e$ is negligible in $\lambda$. Then the protocol described in Procedure~\ref{proc:lcom} is an extractable polynomial commitment scheme for the space $\mL(q,m,k)$ of $q$-linearized polynomials of $q$-degree less than $k$. In particular, the scheme has completeness and satisfies binding and knowledge soundness.

 {For $t=\Theta(\lambda)$ independent query repetitions, the prover time is $\widetilde O(\lambda k)$ and the verifier time is $\widetilde O(\lambda\sqrt{k})$, both are measured in operations over $\mathbb F_{q^m}$. The proof size is $\tilde{O}(\lambda\sqrt{k})$ elements in $\F_{q^m}$}.
\end{theorem}


\begin{remark}
In our $\LRM$-PCS, as in the original Ligero protocol, we use local coordinate queries to test whether a received word $\bu$ is $e$-close to the Gabidulin code $\cC$. Suppose that $d_{\rk}(\bu,\cC)>e$. Then the claimed codeword $\bc=\enc(\bv)$ also satisfies $d_{\rk}(\bu,\bc)>e$. Since the Hamming distance $d_H(\bu,\bc)\ge d_{\rk}(\bu,\bc)$, the words $\bu$ and $\bc$ differ in more than $e$ coordinates. Therefore, after $t$ independent random coordinate queries, the verifier accepts with probability at most $\left(1-\frac{e}{n}\right)^t$.

Alternatively, one can test whether $d_{\rk}(\bu,\cC)\leq e$ using random $\mathbb F_q$-linear queries. Suppose $d_{rk}(\bu,\bc)>e$. Then the error vector $\bu-\bc$ has
$\mathbb F_q$-rank at least $e+1$. Hence, for a uniformly random vector $\ba=(a_1,\ldots,a_n)\in \mathbb F_q^n$,
\[ \Pr_{\ba}\left[\sum_{i=1}^n a_i u_i=\sum_{i=1}^n a_i c_i\right]=\Pr_{\ba}\left[\sum_{i=1}^n a_i(u_i-c_i)=0\right]\le q^{-(e+1)}.
\]
Thus, after $t$ independent random $\mathbb F_q$-linear queries, the soundness error in this step is at most $q^{-t(e+1)}$. However, this approach requires a commitment scheme for the matrix $U$ that supports openings of arbitrary $\mathbb F_q$-linear row combinations of $U$. A natural implementation is to use an $\F_q$-linear commitment to all rows of $U$, which, roughly speaking satisfies $Com(\sum_{i=1}^n a_iU(i))=\sum_{i=1}^na_iCom(U(i))$. But each linear commitment check contributes soundness error roughly
$1/q$. Hence, after $t$ queries, the overall soundness error is governed by $\max\left\{q^{-t},\,q^{-t(e+1)}\right\}=q^{-t}$, which weakens the benefit of the rank-metric linear-query test. Thus, our PCS still uses coordinate queries, while leaving the $\mathbb F_q$-linear combination queries as a direction for future work.
\end{remark}

In summary, our constructions extend code-based succinct-proof techniques from the Hamming metric to the rank metric. However, one limitation is that the soundness error contains the term $\epsilon_e\leq \frac{10q^{n-1}}{q^m}$, which arises from the proximity gaps of Gabidulin codes. To make this term at most $2^{-\lambda}$, the extension degree $m$ must exceed the code length $n$ by approximately $\frac{\lambda+\log_2 10}{\log_2 q}$. Arithmetic over the resulting extension field
$\mathbb F_{q^m}$ may introduce significant computational overhead,
particularly when $q$ is small. This parameter requirement may limit the scalability of the present constructions. Without further optimization, their computational cost is likely to be more suitable for small or moderate-size instances. Identifying concrete applications in these parameter regimes and evaluating their practical performance remain directions for future work.